\documentclass[
 reprint,
 amsmath,amssymb,
 aps,
]{revtex4-2}

\usepackage{ulem}
\usepackage{graphicx}
\usepackage{dcolumn}
\usepackage{bm}
\usepackage{mathrsfs}
\usepackage{hyperref}
\usepackage{xcolor}
\usepackage{subcaption}
\usepackage{cancel}
\usepackage{tikz}
\usepackage{changes}
\usepackage[percent]{overpic}
\newtheorem{theorem}{Theorem}
\newtheorem{lemma}{Lemma}
\newcommand{\D}{D}

\begin{document}

\title{At Most One Inner Killing Horizon in Stationary Black Holes}

\author{Jing-Peng Ye$^{1}$}
\altaffiliation{Contributed equally to this work}
\author{Yu-Xuan Li$^{2,3}$}
\altaffiliation{Contributed equally to this work}
\author{Run-Qiu Yang$^{1}$}\email{aqiu@tju.edu.cn}
\author{Li Li$^{2,3,4}$}\email{liliphy@itp.ac.cn}

\affiliation{${}^{1}$Center for Joint Quantum Studies and Department of Physics, School of Science, Tianjin University, Yaguan Road 135, Jinnan District, 300350 Tianjin, P.~R.~China}
\affiliation{${}^{2}$Institute of Theoretical Physics, Chinese Academy of Sciences, Beijing 100190, China}
\affiliation{${}^{3}$School of Physical Sciences, University of Chinese Academy of Sciences, Beijing 100049, China}
\affiliation{${}^{4}$School of Fundamental Physics and Mathematical Sciences, Hangzhou Institute for Advanced Study, UCAS, Hangzhou 310024, China}

\begin{abstract} 
Stationary-axisymmetric black holes are the standard theoretical framework for astrophysical black holes, yet their interiors remain largely inaccessible: the absence of a hypersurface-orthogonal timelike Killing vector renders the metric non-diagonal, and frame dragging obstructs the usual techniques for probing the interior. We prove that any stationary-axisymmetric black hole satisfying the classical energy conditions contains at most one nondegenerate inner Killing horizon per connected interior branch, irrespective of its matter content. For compact horizon sections, we apply the Raychaudhuri equation to the congruence normal to constant-time hypersurfaces: the strong energy condition makes the squared lapse subharmonic, and the maximum principle excludes a second inner horizon. For noncompact planar and hyperbolic sections central to holography, the null energy condition instead yields a monotonicity obstruction through a geometric flow we introduce here, the ``inverse lapse flow,'' whose weak continuation via viscosity solutions we also construct. The attractive nature of gravity, encoded in these energy conditions, thus both drives horizon formation and caps the number of inner horizons, revealing a striking dual role of the same focusing mechanism. 
\end{abstract}

\maketitle

\textbf{Introduction.--}
Singularity theorems establish that gravitational collapse inevitably leads to spacetime singularities for ordinary matter, whose attractive nature under gravity is encoded in the classical energy conditions~\cite{Penrose:1964wq,Hawking:1966sx,Hawking:1970zqf}. Cosmic censorship then ensures that these singularities remain hidden behind event horizons~\cite{Penrose:1969pc,Penrose:1999vj}. In this picture, classical matter drives both the onset of collapse and the formation of horizons. Event horizons, however, are global objects: they demarcate the boundary of the exterior region, but they reveal little about the detailed structure of the black hole interior. To probe the interior, one must look instead at inner horizons that separate different interior regions. In stationary spacetimes, such inner horizons are generated by Killing vector fields and are therefore naturally described as Killing horizons. Given that classical matter plays a facilitating role in the formation of event horizons via the singularity theorems as well as cosmic censorship, it may be natural to expect that they also promote the appearance of more Killing horizons in the black hole interior.

However, the result is a little surprising. In the Kerr-Newman solution, even when Maxwell field appears, it has at most one Killing horizon (the Cauchy horizon is itself a Killing horizon) inside its event horizon. 
Such a configuration, a stationary pocket bounded by two consecutive inner horizons (see Fig.~\ref{fig:compact-pocket-single} below), would imply a layered internal structure absent in the standard vacuum solutions. How many horizons can appear inside an astrophysical black hole, especially for approximately stationary-axisymmetric case?
Yet despite its greater generality and physical importance, far less progress has been achieved in the stationary case than in its static counterpart. Recent work has established robust constraints on the number of inner horizons for broad classes of static black holes (see \emph{e.g.}~\cite{Hartnoll:2020rwq,Hartnoll:2020fhc,Cai:2020wrp,An:2021plu,Cai:2021obq,Peng:2026baq,Devecioglu:2021xug,Li:2025mhy,Yang:2021civ}), leveraging the fact that static spacetimes admit a hypersurface-orthogonal timelike Killing vector, which diagonalizes the metric and dramatically simplifies the analysis. Stationary spacetimes, by contrast, lack this key property, introducing a series of technical obstructions that have left the stationary case largely open~\cite{Brihaye_2016,gao2024internalstructurehairyrotating,Le_n_2025,Dias:2021afz}. 
Whether the attractive property of matter will promote or exclude multiple inner Killing horizons in general stationary spacetimes has so far been unclear.

This gap is worth addressing not only for its own sake, but also because it speaks to a deeper conceptual point. Ordinary matter is attractive under gravity, a property encoded in the classical energy conditions. This attractive nature plays a generative role in the singularity theorems: it forces gravitational collapse to produce singularities and, via cosmic censorship, horizons. But does the same focusing mechanism also play a restrictive role? That is, once a stationary black hole has formed, does it impose an upper bound on how many inner horizons its interior can support? If so, the attractive nature of gravity would encode not only the onset of collapse but also the ultimate structural simplicity of the black hole interior. This dual role has not been realized in previous studies.

In this Letter, we prove that for ordinary matter with focusing mechanism encoded in the standard energy conditions, a stationary black hole contains at most one nondegenerate inner Killing horizon on each connected interior branch. For compact horizon sections, the attractive nature (quantified by the strong energy condition, SEC) together with the Raychaudhuri equation renders the squared lapse subharmonic. The maximum principle then forbids a stationary pocket. For noncompact horizon sections,
the analogous focusing condition for null congruences—encoded in the null energy condition (NEC)—together with assumptions on topology and asymptotic behavior of the noncompact sector, yields an obstruction through a monotonicity argument via the inverse lapse flow. The generative role of this attractive nature
is already well known: through the Penrose–Hawking theorems, they drive focusing and horizon formation. What we demonstrate here is that the very same focusing mechanism,  when applied to the stationary interior, becomes incompatible with multiple inner horizons. Thus the same attractive nature of ordinary matter under gravity that triggers horizon formation also caps the number of inner horizons. This duality underscores the remarkable efficiency of the attractive nature of gravity in shaping both the exterior and the interior of black hole spacetimes.

Our result joins a growing body of theorems that uncover universal, topology-informed constraints on stationary black hole spacetimes. In a similar spirit, Cunha and Herdeiro showed that the exterior of a $(3+1)$-dimensional stationary-axisymmetric black hole must contain at least one light ring per rotation sense, independent of the matter content~\cite{Cunha_2020}. While their theorem concerns the exterior photon structure, ours concerns the interior horizon structure. Taken together, these results suggest that stationarity, together with the attractive nature of ordinary matter, imposes universal, surprisingly simple constraints on both sides of the event horizon.

\textbf{Stationary-axisymmetric geometry.--}
We consider a $(d+1)$-dimensional stationary-axisymmetric spacetime $(\mathcal{M},g_{\mu\nu})$ admitting two commuting Killing vectors $\{\xi_{(t)}^\mu,\xi_{(\phi)}^\mu\}$, where $\xi_{(t)}^\mu$ and $\xi_{(\phi)}^\mu$ generate stationarity and axial rotations, respectively, as motivated by the black-hole rigidity theorem~\cite{Hollands_2007}. We use the integral curves of these two Killing vectors as the coordinates $(t,\phi)$.
Assuming the spacetime is ``$t-\phi$'' reflection symmetric~\cite{Cunha_2020}, i.e., the metric remains invariant under the transformation $(t,\phi)\rightarrow(-t,-\phi)$, then a general stationary axisymmetric metric can be written as: 
\begin{equation}
\mathrm{d}s^2=-N(x)^2\text{d}t^2+\Psi(x)^2(\text{d}\phi-\omega(x) \text{d}t)^2+q_{AB}(x)\text{d}x^A \text{d}x^B,
\label{metric}
\end{equation}
with $N$ the lapse, $\omega=-g_{t\phi}/g_{\phi\phi}$ the local frame-dragging angular velocity, $\Psi^2=g_{\mu\nu}\xi_{(\phi)}^\mu\xi_{(\phi)}^\nu$, and $q_{AB}=g_{AB}$. Greek indices $\mu,\nu=0,\ldots,d$ refer to spacetime, and $A,B$ label the $d-1$ directions orthogonal to the two-dimensional Killing-orbit plane. Unlike the static case, the off-diagonal component $g_{t\phi}$ prevents the metric from being diagonalized, making the interior analysis substantially more involved.

\textbf{Killing horizons and stationary pockets.--}
The squared lapse is minus the norm of the local co-rotating vector field $\chi^\mu=\xi_{(t)}^\mu+\omega\xi_{(\phi)}^\mu$, namely $g_{\mu\nu}\chi^\mu\chi^\nu=-N^2$. A Killing horizon $\mathcal H$ is generated by 
\begin{equation}
\chi_\mathcal{H}^\mu=\xi_{(t)}^\mu+\Omega_\mathcal{H}\xi_{(\phi)}^\mu ,
\label{eq:horizon-generator}
\end{equation}
where $\Omega_\mathcal H=\omega|_{\mathcal{H}}$ is the constant angular velocity of the horizon, and the generator is a null Killing vector, i.e., $g_{\mu\nu}\chi_\mathcal H^\mu\chi_\mathcal H^\nu=0$~\cite{Cunha_2020}. Consequently,
\begin{equation}
N^2\big|_{\mathcal H}=-g_{\mu\nu}\chi_\mathcal{H}^\mu\chi_\mathcal{H}^\nu=0\,.
\label{eq:lapse-horizon}
\end{equation}
To see the converse, note that in the metric~\eqref{metric}, $N^2=-g_{\mu\nu}\chi^\mu\chi^\nu$. On a connected component of $N^2=0$, $\chi^\mu$ is null and orthogonal to both $\xi_{(t)}^\mu$ and $\xi_{(\phi)}^\mu$, which are tangent to the hypersurface. The weak rigidity theorem then implies $\omega$
is constant on that component~\cite{Carter:1969zz,Straumann:2013spu}; hence $\chi^\mu$ is a Killing vector on that component and the hypersurface is a Killing horizon. Thus, the hypersurface $\mathcal{H}$ is a Killing horizon if and only if $N^2|_{\mathcal{H}}=0$,  see Supplemental Material (SM Sec.\,S1) for more details\,\footnote{The Supplemental Material, which includes Refs.~\cite{Straumann:2013spu,Jang:1977kef,https://doi.org/10.1111/j.1749-6632.1973.tb41445.x,evans2022partial,10.1093/ptep/ptx072}, presents the technical details supporting the results made in the main text.}.

For the purpose of rigour, the black-hole interior is defined as the region reached by future-directed ingoing null geodesics that emanate from a single connected component of the event horizon~\cite{Yang:2021civ}. 
Using this definition, the Schwarzschild black hole has no inner Killing horizon, while the Kerr-Newman black holes have at most one inner Killing horizon, agreeing with the usual statement in the literature.
Suppose that two consecutive nondegenerate inner Killing horizons $\mathcal H_1$ and $\mathcal H_2$ occur on the same black hole interior branch, as in Fig.~\ref{fig:compact-pocket-single}. We call the region bounded by the two inner horizons a stationary pocket $\mathcal V$. On the constant-$t$ hypersurface $\Sigma_t$ for Eq.~\eqref{metric}, the spatial section of the pocket $\mathcal D:=\Sigma_t\cap\mathcal V$ obeys
\begin{equation}
N^2=0\quad\hbox{on}\quad\partial\mathcal D,\qquad
N^2>0\quad\hbox{in}\quad{\mathring{\mathcal D}}.
\label{eq:pocket}
\end{equation}
Here, $\partial\mathcal D=(\mathcal H_1\cup\mathcal H_2)\cap\Sigma_t$, and $\mathring{\mathcal D}$ denotes the interior of $\mathcal D$. Note that, in this region, $g^{\mu\nu}(\mathrm{d}t)_{\mu}(\mathrm{d}t)_{\nu}=-N^{-2}<0$, so $\Sigma_t$ is spacelike inside $\mathcal{D}$.
\begin{figure*}[tbph]
    \centering
    \captionsetup{
        justification=raggedright,
        singlelinecheck=false
    }
    \captionsetup[subfigure]{justification=centering}
    \begin{subfigure}[b]{0.44\textwidth}
        \centering
        \phantomsubcaption
        \label{fig:compact-pocket-single}
        \raisebox{0.55cm}{
            \begin{overpic}[
                width=0.8\linewidth
            ]{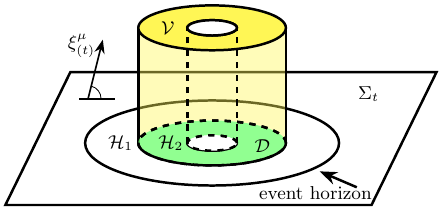}
                \put(2,42){
                    \makebox(0,0)[lt]{\textbf{(\thesubfigure)}}
                }
            \end{overpic}
        }
    \end{subfigure}
    \quad
    \begin{subfigure}[b]{0.46\textwidth}
        \centering
        \phantomsubcaption
        \label{fig:maximum principle}
        \begin{overpic}[
            width=0.8\linewidth
        ]{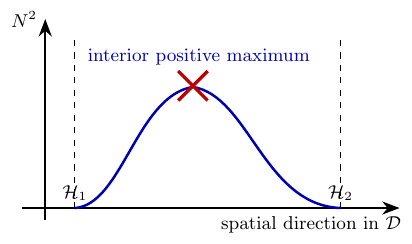}
            \put(-2,49){
                \makebox(0,0)[lt]{\textbf{(\thesubfigure)}}
            }
        \end{overpic}
    \end{subfigure}
    \caption{(a) A hypothetical stationary pocket $\mathcal{V}$ (yellow) bounded by two consecutive inner Killing horizons $\mathcal{H}_1$ and $\mathcal{H}_2$ on a connected interior branch. Its spatial section $\mathcal{D} = \Sigma_t \cap \mathcal{V}$ is compact. (b) Under the SEC, $N^2$ is subharmonic in the pocket $\mathcal{V}$; the maximum principle forces $N^2 \le 0$ throughout $\mathcal{D}$, contradicting the required $N^2>0$ in $\mathcal{V}$.}
\end{figure*}

\textbf{Raychaudhuri equation and focusing mechanism.--}
The Raychaudhuri equation encodes the attractive nature of ordinary matter under gravity by governing the focusing of geodesic congruences. To analyze the key feature of the pocket $\mathcal V$, let us first review the basis of Raychaudhuri equation in general setup. Consider a general timelike congruence, which is initially hypersurface-orthogonal to a co-dimensional one hypersurface $\Sigma_{t}$. Denote $h_{\mu\nu}$ to be the induced metric of $\Sigma$, $\nabla_\mu$ to be the spacetime covariant derivative, and $n^\mu$ to be future-directed unit normal of $\Sigma$. 
The Raychaudhuri equation in this setup reads~\cite{Wald:1984rg,Abreu_2011}
\begin{equation}
n^\mu\nabla_\mu\theta
=-\frac{1}{d}\theta^2-\sigma_{\mu\nu}\sigma^{\mu\nu}
-R_{\mu\nu}n^\mu n^\nu+\nabla_\mu a^\mu ,
\label{eq:raychaudhuri}
\end{equation}
where $d$ is the number of spatial dimensions, $\theta=\nabla_\mu n^\mu$ is the expansion, $a^\mu=n^\nu\nabla_\nu n^\mu$ is the acceleration, and $\sigma_{\mu\nu}=h_\mu{}^\alpha h_\nu{}^\beta\nabla_{(\alpha}n_{\beta)}
-\frac{1}{d}\theta h_{\mu\nu}$
is the shear tensor. The shear tensor satisfies $\sigma_{\mu\nu}\sigma^{\mu\nu}\geq 0$. 

Consider the trace-reversed form of Einstein’s equation $R_{\mu\nu}=T_{\mu\nu}-\frac{1}{d-1}Tg_{\mu\nu}$ for which we have set $8\pi G=1$ and have absorbed the cosmological constant contribution into the energy momentum tensor $T_{\mu\nu}$. The right-hand side is precisely the combination constrained by the SEC, which requires $\left(T_{\mu\nu}-\frac{1}{d-1}Tg_{\mu\nu}\right)v^{\mu}v^{\nu}\ge0$ for any timelike vector $v^{\mu}$. Since the normal vector $n^{\mu}$ is timelike, the SEC implies $R_{\mu\nu}n^{\mu}n^{\nu}\ge0$.
Therefore, both the shear term and the Ricci term decrease the expansion and therefore promote focusing. The Raychaudhuri equation~\eqref{eq:raychaudhuri} thus captures the physical fact that ordinary matter is attractive under gravity: for a hypersurface-orthogonal geodesic congruence (i.e., with vanishing acceleration $a^{\mu}=0$), the SEC guarantees that an initially negative expansion diverges to negative infinity within finite proper time, signalling the formation of a focal or conjugate point. For null congruences, the NEC plays an analogous role through the null Raychaudhuri equation. When combined with the additional global causal and trapping assumptions of the singularity theorems, this local focusing mechanism leads to timelike or null geodesic incompleteness~\cite{Penrose:1964wq,Hawking:1966sx,Hawking:1970zqf}. Thus, through the focusing encoded in the Raychaudhuri equation, the attractive nature of ordinary matter---quantified by the SEC for timelike and the NEC for null congruences---promotes both the formation of singularities and, via cosmic censorship, the appearance of horizons.

\textbf{Maximum principle and at most one inner Killing horizon.--}The Raychaudhuri equation, together with the SEC, drives focusing and horizon formation. We now show that this same focusing mechanism also forbids the appearance of the second inner Killing horizon.

Instead of applying Eq.~\eqref{eq:raychaudhuri} to timelike geodesic congruence, we now apply it to the timelike curves generated by unit normal covector $n_\mu=-N\nabla_\mu t$ of constant-$t$ hypersurface $\Sigma_t$ for stationary spacetime~\eqref{metric}. The corresponding unit normal vector reads
\begin{align}    
n^{\mu}=N^{-1}\left(\xi^{\mu}_{(t)}+\omega\xi_{(\phi)}^\mu\right).
\end{align}
Using the Killing identities $\nabla_{\mu}\xi^{\mu}_{(t)}=0$ and $\nabla_{\mu}\xi^{\mu}_{(\phi)}=0$, along with the metric \eqref{metric}, we obtain $\theta=\nabla_{\mu}n^{\mu}=0$. It should be emphasized that while this congruence is hypersurface-orthogonal, it is not geodesic. Furthermore, the acceleration vector satisfies $a_\mu=\partial_\mu\ln N$, and its divergence is $\nabla_\mu a^\mu=N^{-1}D^2N$, where $D^2:=D_a D^a$ and $D_a$ is the induced covariant derivative on the hypersurface $\Sigma_t$. Here and throughout, lowercase Latin indices $a,b=1,\dots,d$ label the $d$-dimensional spatial coordinates on the constant-$t$ hypersurfaces $\Sigma_{t}$. Substituting these relations into Eq.~\eqref{eq:raychaudhuri} yields
\begin{equation}
\D^2N=N\left(\sigma_{\mu\nu}\sigma^{\mu\nu}+R_{\mu\nu}n^\mu n^\nu\right).
\label{eq:lapse}
\end{equation}
The SEC then ensures
\begin{equation}
\D^2N^2=2N^2\left(\sigma_{\mu\nu}\sigma^{\mu\nu}+R_{\mu\nu}n^\mu n^\nu\right)+2(\D_a N)(\D^a N)\ge 0.
\label{eq:subharmonic}
\end{equation}
Thus $N^2$ is subharmonic on the compact domain $\mathcal D$ (the green region of Fig.~\ref{fig:compact-pocket-single}). By the maximum principle, any subharmonic function attains its maximum on the boundary. Hence $\max_{\mathcal D}N^2=\max_{\partial\mathcal D}N^2=0$, since $N^2=0$ on both boundary components. This contradicts the requirement $N^2>0$ in the interior of a stationary pocket~\eqref{eq:pocket} (see Fig.~\ref{fig:maximum principle}). Therefore the SEC and Raychaudhuri equation exclude a compact stationary pocket. Thus the same focusing mechanism—attractive nature quantified by the SEC via the Raychaudhuri equation—that drives gravitational collapse and horizon formation also forbids a second inner horizon.

\textbf{Beyond compact horizon sections.--} 
In holographic duality, noncompact horizon geometries are of particular interest. Planar AdS black holes constitute the standard framework for describing strongly coupled quantum systems via holography~\cite{hartnoll2018holographicquantummatter,Cai:2015cya,Liu:2020rrn,Baggioli:2021xuv}. Hyperbolic horizons, meanwhile, play a fundamental role in topological black hole thermodynamics and holographic entanglement, rather than arising as incidental special solutions~\cite{Birmingham_1999,Casini_2011,Hung_2011}. For such cases, the maximum-principle argument employed previously is no longer applicable. Furthermore, in holographic applications the NEC serves as the relevant constraint: it controls bulk causal structure~\cite{Gao_2000} and underpins the strong subadditivity of holographic entanglement entropy~\cite{Callan_2012,Wall_2014}. Nevertheless, the underlying physical mechanism remains the same: the attractive nature of ordinary matter, quantified by the NEC, prevents the existence of a stationary pocket.  

Assume that two nondegenerate inner Killing horizons bound a stationary pocket. 
Let a smooth function $z$ foliate $\mathcal D$ by nested, noncompact leaves $\mathcal S_z:=\{z=\mathrm{const.}\}$ (see Fig~\ref{fig}), with $\mathcal S_{z_1}$ and $\mathcal S_{z_2}$ approaching the two inner horizon sections located at $z=z_1$ and $z=z_2>z_1$, respectively. 
Focusing on applications in holography rather than aimlessly mathematical generalizations, we here restrict ourselves to considering the case that horizons have planar $\mathbb{R}^{d-1}$ or hyperbolic $\mathbb{H}^{d-1}$ topology (may not have planar or hyperbolic symmetry). The topology of $\mathcal{D}$ is just the simple direct product of $\mathbb{R}$ and topology of the horizon. In addition, in order to ensure well-controlled asymptotic behavior of $\mathcal{D}$ along the noncompact directions, we assume that the deviation from planar or hyperbolic symmetry only occurs in some compact supports of $\mathcal{D}$ and so $\mathcal{D}$ has asymptotically planar or hyperbolic symmetry along noncompact directions. This could be regarded as the generalization of planar or hyperbolic AdS black holes in non-homogeneous systems and has been widely considered in holographic models.

Let $s^a$ be the unit normal to $\mathcal S_z$ within $\Sigma_t$, oriented by $s^aD_a z>0$, let $\gamma_{IJ}$ be the induced metric on $\mathcal{S}_z$ and $\mathcal K=\gamma^{IJ}\mathcal{K}_{IJ}$ be the trace of the extrinsic curvature of $\mathcal S_z\subset\Sigma_t$. We take the compact subset $\mathfrak{B}_{z,L}\subset \mathfrak{B}_z$ such that $\mathfrak{B}_{z,L}= \mathfrak{B}_{z}$ as $L\rightarrow\infty$, and define
\begin{equation}
  \begin{aligned}
    F(z):&=\lim_{L\rightarrow\infty}
    \frac{1}{\mathcal{A}(L)}\int_{\mathfrak B_{z,L}}
    \left[s^a\partial_aN-\frac{\mathcal K N}{d-1}\right]\mathrm dS  \\
    &=\Big\langle s^a\partial_aN-\frac{\mathcal K N}{d-1}\Big \rangle_z,
  \end{aligned}
\label{F}
\end{equation}
where $\mathcal{A}(L):=\int_{\mathfrak B_{z_1,L}}\mathrm dS$ is area of $\mathfrak B_{z_1,L}$ with a fixed $z=z_1$ and we write $\langle X\rangle_z:=\lim_{L\to\infty}\mathcal A^{-1}(L)\int_{\mathfrak B_{z,L}}X\,\mathrm dS$ as the average density of $X$ on $\mathfrak B_{z,L}$. Note we used $z$-independent ``area'' $\mathcal{A}(L)$ to obtain the average density. It is introduced to regularize the divergent integrals on noncompact surfaces $\mathfrak{B}_{z}$.
%
%
%
Some variants of function $F(z)$ appeared in the no-inner-horizon theorems~\cite{Hartnoll:2020fhc,Cai:2020wrp,An:2021plu,Cai:2021obq}, as part of the radial conserved charge. We will show that this quantity serves as a monotonic ‘charge’ along the flow.

We note that Eq.~\eqref{eq:lapse-horizon} does not depend on the topology of the horizon section, because it relies only on the null property of the horizon. Since $N=0$ on both horizons, the $\mathcal K N$ term vanishes there, and $s^a\partial_aN$ becomes the leading contribution $F(z)\sim \langle s^a\partial_a N\rangle_{z}$. Note that $F(z_2)$ and $F(z_1)$ are both independent of the choice of foliation, though the other values of $F(z)$ depend on foliation. Since $N$ is positive between $z_1$ and $z_2$, the chosen orientation implies that $s^a\partial_aN>0$ on the first horizon (where $N$ increases into the pocket) and $s^a\partial_aN<0$ on the second horizon (where $N$ decreases toward the boundary). Thus $F(z_1)>0$ while $F(z_2)<0$. In the following, 
we will use NEC to show $F(z_1)\leq F(z_2)$, which gives a contradiction. Hence the noncompact stationary pocket cannot exist.


\begin{figure}
	\includegraphics[width=7cm]{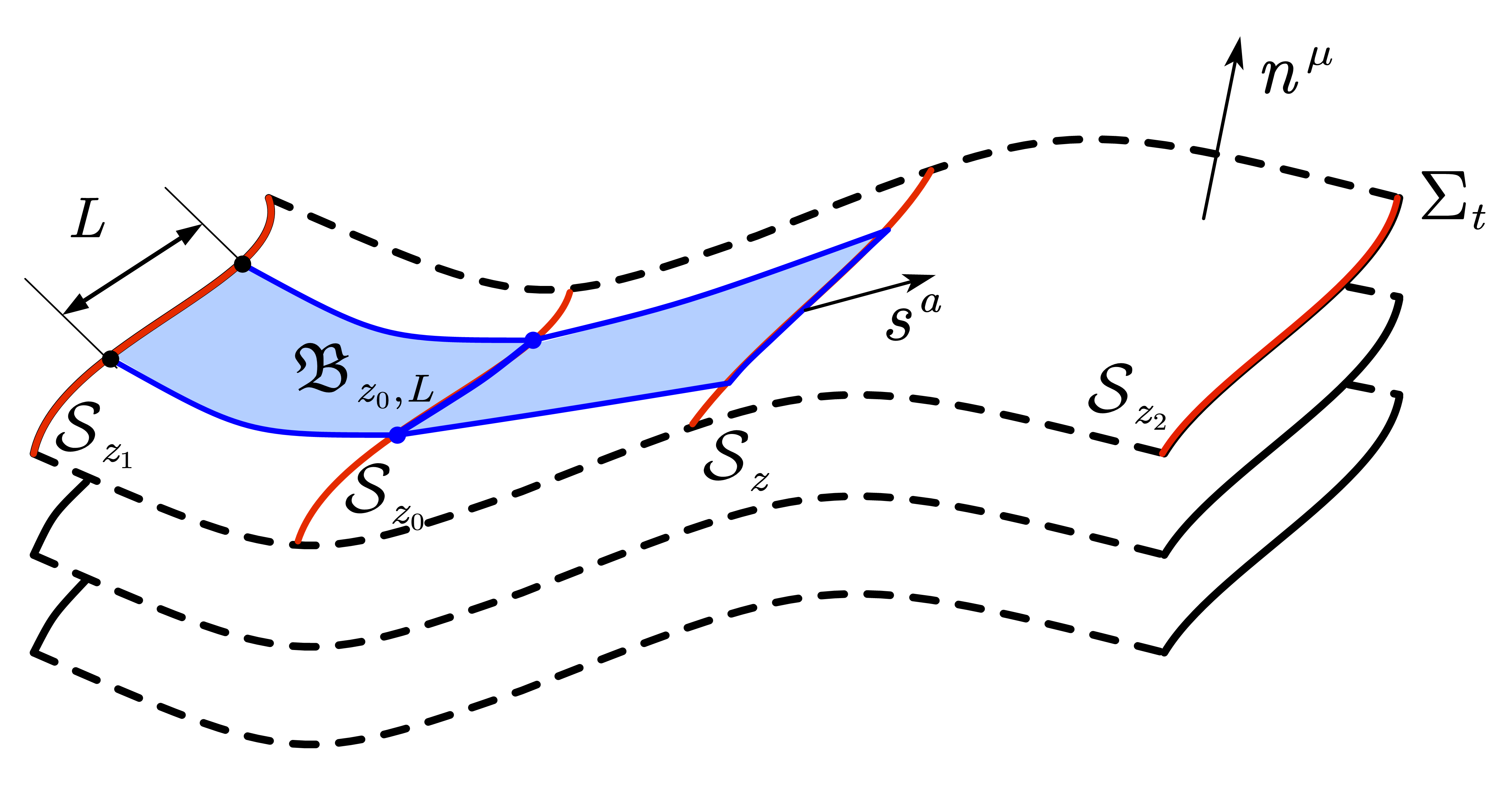}
    \captionsetup{
    justification=raggedright,
    singlelinecheck=false
    }
	\caption{
    Foliation of the noncompact spatial section $\mathcal{D}$ by leaves $\mathcal{S}_z$ (red) between the two inner horizons at $z=z_1$ and $z=z_2$. The compact cutoff $\mathfrak{B}_{z,L}$ (blue) defines the averaged flux $F(z)$ in Eq.~\eqref{F}; its $L\to\infty$ limit yields a monotonic quantity that contradicts the boundary signs required for a stationary pocket.}
	\label{fig}
\end{figure}

To establish this, we now use the Gauss--Codazzi relation and the Einstein constraint equations to obtain the derivative of $F(z)$ (see SM Sec.\,S2 for details):
\begin{equation}
\begin{aligned}
F'(z) &=\frac{1}{d-1} \Big\langle  N\varphi(\varpi - \mathfrak{R}) + \mathfrak{D}^2(N\varphi)\\
&\qquad - 2\gamma^{IJ}(\partial_I\varphi)(\partial_J N)\Big \rangle_z\\
&+\frac{1}{d-1}\Big\langle\frac{\Psi^2\varphi}{N}\left[\frac{d}{2}(s^a \partial_{a}\omega)^2+\frac{d-1}{2}(\mathfrak{D}\omega)^2 \right]\Big \rangle_z\, ,
\label{eq:fluxidentity}
\end{aligned}
\end{equation}
where $\varphi$ is the lapse function of the foliation $\{\mathcal{S}_z\}$, the $N$,$\Psi$ and $\omega$ are defined in metric~\eqref{metric}, the $\mathfrak D$ and $\mathfrak R$ are covariant derivative and scalar curvature of $\mathcal S_z$, $\varpi$ denotes the matter combination
\begin{equation}
\varpi=\frac{1}{2}\sum_{I=1}^{d-1}\left(T_{\mu\nu}\ell_{I,+}^\mu\ell_{I,+}^\nu+T_{\mu\nu}\ell_{I,-}^\mu\ell_{I,-}^\nu\right),
\label{eq:varpi}
\end{equation}
in which $\ell_{I,\pm}^\mu=n^\mu\pm e_I^\mu$. Here $\{e_I^\mu\}$ is unit orthogonal basis on $\mathcal{S}_z$, so that $\{s^\mu,e_1^\mu,e_2^\mu\,\dots,e_{d-1}^\mu\}$ constitute the spatial unit basis in the orthogonal decomposition of $\Sigma_t$.  For each $I$, $n^\mu$ is unit timelike, $e_I^\mu$ is unit spacelike, and they are orthogonal, so $\ell_{I,+}^\mu$ and $\ell_{I,-}^\mu$ are null vectors. The NEC directly gives $\varpi\ge0$ because $T_{\mu\nu}\ell^\mu\ell^\nu\ge0$ for every null vector $\ell^\mu$.

To determine the sign of Eq.~\eqref{eq:fluxidentity}, we need to consider the remaining terms, which involve the lapse function $\varphi$. We choose the gauge $\varphi N=1$ and call foliation under such gauge as ``inverse lapse flow.'' This is an analogue of Geroch's inverse mean curvature flow $\varphi \mathcal K=1$~\cite{Jang:1977kef,https://doi.org/10.1111/j.1749-6632.1973.tb41445.x,huisken2001inverse}.
Here we adopt the similar strategy but with the lapse $N$ playing the role of $\mathcal K$. To our knowledge, this is the first use of the inverse lapse function as the generator of a normal flow in the context of black hole physics. Under the inverse lapse flow $\varphi N=1$, Eq.~\eqref{eq:fluxidentity} becomes 
\begin{equation}\label{eq:Fprimegauge}
\begin{aligned}
F'(z)&=\frac{1}{d-1}\Big\langle
\varpi-\mathfrak R+2N^{-2}\gamma^{IJ}(\partial_I N)(\partial_J N)
\Big\rangle_z \\
&+\frac{1}{d-1}\Big\langle\frac{\Psi^2}{N^2}
\Big[\frac{d}{2}(s^a\partial_a\omega)^2 
+\frac{d-1}{2}(\mathfrak D\omega)^2\Big]\Big\rangle_z\,.
\end{aligned}
\end{equation}
For the noncompact case, two typical examples  planar $(k=0)$ and hyperbolic $(k=-1)$ symmetric black holes. In these two cases, $\mathfrak{R}=(d-1)(d-2)kz^2$, which gives $\langle \mathfrak{R}\rangle_z\le0$. 
Without above symmetries, one can show that $\langle \mathfrak{R}\rangle_z \le 0$ still holds for the case considered here when $d \le 3$ due to the topology of horizons. For $d >3$, the value of $\langle \mathfrak{R}\rangle_z$ is not determined solely by the topology. However, for the physically relevant class that $\mathcal{D}$ has asympototical planar or hyperbolic symmetry along noncompact directions, any compact regions of positive curvature contribute subdominantly to the infinite-volume average. Consequently, $\langle \mathfrak{R}\rangle_z \le 0$ holds (See SM Sec.\,S2.2).
%
%
Combining with $\varpi\ge0$ required by the NEC, all terms in right-hand of Eq.~\eqref{eq:Fprimegauge} are nonnegative. Hence $F'(z)\ge0$ throughout the pocket under the foliation of inverse lapse flow. 

If above inverse lapse flow can give a smooth foliation on the pocket, we then find $F(z_2)\geq F(z_1)$.  Since the values of $F(z_1)$ and $F(z_2)$ are independent of the choice of foliation, we then find $F(z_2)\geq F(z_1)$ will be true for all foliations, contradicting $F(z_2)< F(z_1)$ required by existence of pocket. The noncompact stationary pocket is therefore excluded for this class. Like the inverse mean curvature flow, the inverse lapse flow also suffers from divergency at horizon (since $N=0$ at horizon and so $\varphi N=1$ losing validity) and 
may also develop singularities in the intermediate process. In addition, it may suffer from overdetermined problem due to the fact that the flow starts from one horizon and is required to reach the other horizon simultaneously. These key issues are all handled by the viscosity-solution framework for the eikonal equation~\cite{evans2022partial} (see SM Sec.\,S3 for details) and so the conclusion generally holds. The attractive nature of ordinary matter, as encoded in the NEC, plays an essential role in this argument: it ensures $\varpi \ge 0$, which, together with $\langle\mathfrak{R}\rangle_z \le 0$, renders $F'(z)\ge 0$ and hence makes the monotonicity argument possible.


\textbf{Discussion.--}
We have shown that ordinary matter with its attractive nature under gravity encoded in the classical energy conditions not only guarantees the formation of horizons but also limits the number of inner horizons in the interior of a stationary black hole.
This extends the static results of~\cite{Yang:2021civ}—which relied crucially on the diagonal form of the metric—to the more physically relevant stationary case. More importantly, it uncovers that the Raychaudhuri equation plays an essential role both in formation of event horizon and restriction on inner horizons.

Our result is purely geometric and relies only on the stationarity of the interior region. It applies to the physically most relevant class of black holes, namely stationary and axisymmetric ones, independent of the matter content. Inside the black hole, quantum effects of matter and gravity may lead to violations of energy conditions and Einstein gravity. But such violations are expected in strong-field and small spacetime regions. The theorem remains robust as long as the horizon scale is well above the quantum-gravity regime, where classical energy conditions and Einstein gravity hold. Though we here consider the stationary case, the results also raise interesting inspirations on dynamical process of black hole formation. While multiple horizons may appear in dynamical collapse~\cite{cao2016multihorizoncriticalbehaviorgravitational}, our result implies that, if they do, they must merge or be destroyed before the system settles into a stationary final state. Our proof for noncompact case needs an additional asymptotic symmetry when $d>3$. Extending the noncompact proof to fully general asymptotics remains a mathematically open question, though no physical counterexample is currently known.

Our findings open several avenues for future work. First, it would be instructive to determine whether analogous horizon-capacity bounds survive when quantum effects or higher-derivative corrections become important~\cite{duan2026kasnerepochsorderedoscillations}.
Second, the radial monotonicity of $F(z)$ also holds outside the event horizon, and so bears a striking formal resemblance to holographic c-functions: its monotonic evolution along the radial direction may encode the loss of degrees of freedom in a dual renormalization-group flow, a connection worth exploring in explicit AdS/CFT setups. Finally, the remaining single inner Killing horizon—when present—constitutes a Cauchy horizon. Our theorem shows that the energy conditions alone cannot eliminate it; whether additional, physically motivated restrictions on the matter could rule it out remains a decisive question, as its removal would provide a direct geometric underpinning for strong cosmic censorship.


%
\textbf{Acknowledgement.--}
This work is supported by the National Natural Science Foundation of China Grants No.\,12525503, No.\, 12375051, No.\,12588101 and No.\,12447101, and by the Tianjin University Self-Innovation Fund Extreme Basic Research Project Grant No.\, 2025XJ22-0014 and No.\, 2025XJ21-0007.

\makeatletter
\renewcommand{\bibsection}{%
  \begingroup
    \let\addcontentsline\@gobblethree
    \section*{\refname}%
  \endgroup
  \@nobreaktrue
}
\makeatother

\bibliography{prl_references}

\clearpage

\cleardoublepage

\setcounter{equation}{0}
\renewcommand{\theequation}{S\arabic{equation}}
\setcounter{figure}{0}
\renewcommand{\thefigure}{S\arabic{figure}}
\setcounter{section}{0}
\renewcommand{\thesection}{S\arabic{section}}
\renewcommand{\thesubsection}{\thesection.\arabic{subsection}}

\begin{center}
\textbf{\Large Supplemental Material}
\end{center}

\noindent In this Supplemental Material, we provide technical details supporting the results reported in the main Letter.

\tableofcontents

\section{Killing horizon in stationary axisymmetric spacetime}\label{SeqLiou}

In a stationary axisymmetric spacetime, there exist two linearly independent Killing vector fields: a Killing vector field $\xi_{(t)}^{\mu}$, which is timelike at infinity, representing time-translation symmetry, and a spacelike Killing vector $\xi_{(\phi)}^{\mu}$, representing rotational symmetry, whose integral curves are closed. Furthermore, these two Killing vectors commute with each other, i.e., $[\xi_{(t)},\xi_{(\phi)}]=0$. Typically, we choose the integral curves of $\xi_{(t)}^{\mu}$ and $\xi_{(\phi)}^{\mu}$ as the $t$ and $\phi$ coordinate axes, respectively. Assuming the spacetime is ``$t-\phi$" reflection symmetric, i.e., the metric remains invariant under the transformation $(t,\phi)\rightarrow(-t,-\phi)$, then under these assumptions, a general stationary axisymmetric metric can be written in the form: 
\begin{align} 
    ds^2 = -N^2\text{d}t^2 + \Psi^2 (\text{d}\phi - \omega \text{d}t)^2 + q_{AB} \text{d}x^{A} \text{d}x^{B}, \label{metricAd}
\end{align}
where $g_{\phi\phi} = \Psi^2 > 0$ and $A, B = 1, 2, ..., d-1$. And $N$ is the lapse function and $\beta^{\mu} = -\omega (\partial / \partial \phi)^\mu$ is the shift vector. Next we ask what a Killing horizon implies in such a spacetime. In the following, we prove a theorem.
\begin{theorem}
    In a stationary axisymmetric spacetime, a hypersurface $\mathcal{H}$ is a Killing horizon if and only if lapse function $N^2$ vanishes on $\mathcal{H}$.
\end{theorem}
\textit{Proof:} We first prove the necessity of the theorem. If $\mathcal{H}$ is a Killing horizon, there exists a constant $\Omega_H$ such that the Killing vector associated with $\mathcal{H}$ can be written as $\chi^\mu = \xi_{(t)}^{\mu} + \Omega_H \xi_{(\phi)}^\mu$ (i.e., $\chi^\mu$ is the generator of $\mathcal{H}$). Because $\xi^\mu_{(t)}$ and $\xi^\mu_{(\phi)}$ are killing vector with $[\xi^\mu_{(t)},\xi^\mu_{(\phi)}]=0$, we have
\begin{equation}
    \mathcal{L}_Y(\chi^\mu\chi_\mu)=(\mathcal{L}_Y g_{\mu\nu})\chi^\mu\chi^\nu+2g_{\mu\nu}(\mathcal{L}_Y\chi)^\mu\chi^\nu=0\,,
\end{equation}
where $\mathcal{L}$ is Lie derivative, and $Y^\mu\in\{\xi^\mu_{(t)},\xi^\mu_{(\phi)}\}$, which means $\mathcal{L}_Y$ induces a single-parameter isometric group action $\sigma_s^Y(\mathcal{H})=\mathcal{H}$, and differentiate at point $s=0$ gives:
\begin{equation}
    \{\xi^\mu_{(t)}|_\mathcal{H},\xi^\mu_{(\phi)}|_\mathcal{H}\}\in T\mathcal{H}\,,
\end{equation}
where $T\mathcal{H}$ is the tangent space of $\mathcal{H}$. Since $\chi^\mu$ is normal to $\mathcal{H}$ and $\xi_{(t)}^{\mu},\xi_{(\phi)}^\mu$ are tangent on $\mathcal{H}$, we have
\begin{align}
    &\chi_\mu \xi_{(t)}^{\mu} = \chi_\mu \xi_{(\phi)}^\mu = 0\notag \\ \Rightarrow \quad &
    \xi_{(t)}^{2} + \Omega_H \xi_{\mu(\phi)}\xi_{(t)}^{\mu} = \xi_{\mu(t)}\xi_{(\phi)}^{\mu} + \Omega_H \xi^2_{(\phi)}=0\,.
\end{align}
Eliminating $\Omega_{H}$, we therefore obtain
\begin{align}
&\xi_{(t)}^{2} \xi_{(\phi)}^2|_{\mathcal{H}} = (\xi_{\mu(\phi)}\xi_{(t)}^{\mu})^2|_{\mathcal{H}} = 0\notag \\  \Rightarrow \quad & (g_{tt} g_{\phi\phi} - g_{t\phi}^2)|_{\mathcal{H}} = 0. \label{det}
\end{align}
Using the metric ansatz~(\ref{metricAd}), the lapse function satisfies
\begin{align}\label{N2}
-N^2 = \frac{g_{tt}g_{\phi\phi} - g_{t\phi}^2}{g_{\phi\phi}}.
\end{align}
Since $\xi_{(\phi)}^\mu$ is spacelike, $g_{\phi\phi}>0$. Therefore, when $\mathcal{H}$ is a Killing horizon, Eq.~\eqref{det} implies $N^2|_{\mathcal{H}}=0$.

Next, we prove the sufficiency of the theorem. From Eq.~\eqref{N2}, it can be seen that  $N^2$ is the norm of the co-rotating vector $\chi^{\mu} = \xi_{(t)}^{\mu} + \omega \xi_{(\phi)}^{\mu}$, (i.e. $-N^2 = \chi^{\mu} \chi_{\mu}$). Due to the symmetry of the metric, the Killing vectors $\xi_{(t)}^{\mu}$ and $\xi_{(\phi)}^{\mu}$ are necessarily tangent to $\mathcal{H}$. Therefore, the coordinates $t$ and $\phi$ also lie on $\mathcal{H}$. To further analyze the properties of this hypersurface, we define a coordinate system $\{t, \phi, y^{\tilde A}, \tilde A = 1, 2, ..., d-2 \}$ on the hypersurface $\mathcal{H}$. The new coordinate basis vectors in the original coordinate representation are:
\begin{align}
    \left( \frac{\partial}{\partial y^{\tilde A}} \right)^{\mu} = {C_{\tilde A}}^B \left( \frac{\partial}{\partial x^{B}} \right)^{\mu}.
\end{align}
Since $g_{\mu\nu}\xi_{(t)}^{\mu} (\partial / \partial x^A)^\nu =g_{\mu\nu} \xi_{(\phi)}^{\mu} (\partial / \partial x^A)^\nu = 0$, it follows that $\chi_\mu (\partial / \partial y^{\tilde A})^\mu = 0$. Meanwhile, from Eq.~\eqref{det}, $N^2=0$ implies $\chi_{\mu}\xi^{\mu}_{(t)}=0$ and $\chi_{\mu}\xi^{\mu}_{(\phi)}=0$. This means the vector field $\chi_\mu$ is both orthogonal to and tangent to $\mathcal{H}$. Therefore, $\mathcal{H}$ is a null hypersurface. To prove that $\mathcal{H}$ is a Killing horizon, we only need to show that the vector field $\chi_{\mu}$ is a Killing vector. Here, we refer a theorem
 (the Theorem 8.3 (Weak rigidity theorem) in Ref.~\cite{Straumann:2013spu}):
\begin{lemma}
    Let $(\mathcal{M},g)$ be a circular spacetime with commuting Killing vector fields $\xi_{(t)}^\mu$ and $\xi_{(\phi)}^\mu$, and let $\omega = -\xi^\mu \Psi_\mu / \Psi^2$ and $\chi^\mu = \xi_{(t)}^\mu + \omega \xi_{(\phi)}^\mu$. Then the hypersurface $\mathcal{H}$ on which $\chi_\mu \chi^\mu = 0$ (assumed to exist) is null, and $\omega$ stays constant on $\mathcal{H}$.
\end{lemma}
Therefore, according to this theorem, we see that the scalar field $\omega$ is constant on $\mathcal{H}$. This indicates that the vector field $\chi_{\mu}$ is a Killing vector, and thus $\mathcal{H}$ is a Killing horizon. 

\section{Nonnegativity of $F'(z)$}\label{S2}
\subsection{Asymptotic symmetry along transverse directions}\label{Sec2.1}

Before we give the details on our discussions for noncompact case, let us first clarify the asymptotic assumption that underlies the definition of the averaged flux $F(z)$ and the sign of $\langle \mathfrak R\rangle_z$ in the noncompact case. As illustrated in Fig.~\ref{Asymptotic}, the geometry of the stationary pocket $\mathcal{D}$ involves two distinct directional hierarchies: the bulk radial direction runs from the AdS boundary toward the interior (along which the foliation parameter $z$ is defined), while the transverse noncompact directions extend to spatial infinity tangent to the horizon sections. The key assumption is that any deviation from exact planar ($\mathbb{R}^{d-1}$) or hyperbolic ($\mathbb{H}^{d-1}$) symmetry is confined to a compact support in the transverse directions. Thus, sufficiently far along the noncompact directions at any fixed $z$ surface, the induced metric on $\mathcal{S}_z$ approaches the reference homogeneous geometry. Consequently, as depicted by the blue shaded region in Fig.~\ref{Asymptotic}, the transverse asymptotic regime is effectively the direct product of the bulk radial direction with either $\mathbb{R}^{d-1}$ or $\mathbb{H}^{d-1}$. 

\begin{figure}[tbph]
	\centering
	\includegraphics[width=0.48\textwidth]{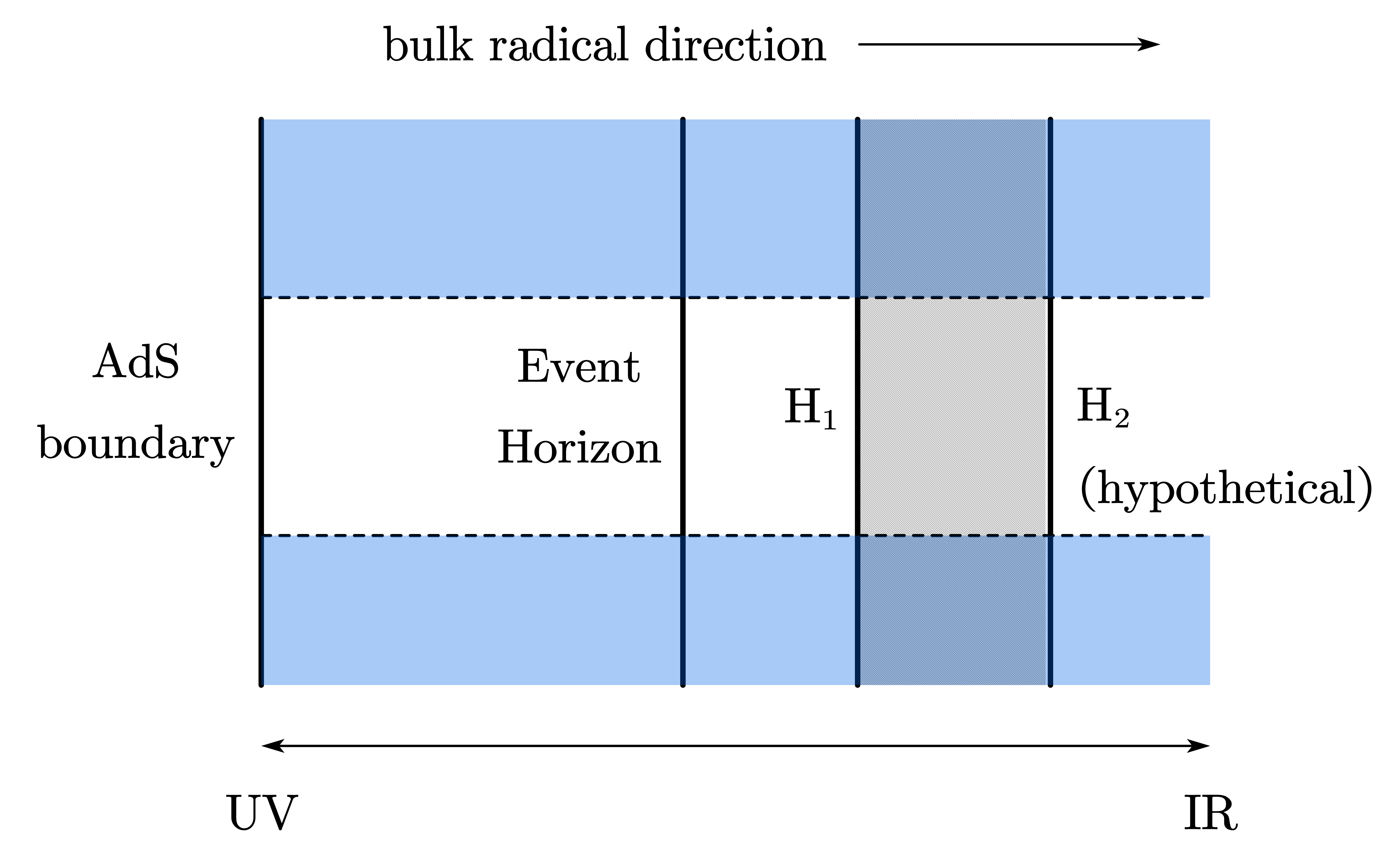}
    \captionsetup{
    justification=raggedright,
     singlelinecheck=false
    }
	\caption{Schematic diagram on asymptotic symmetry along transverses directions. The blue region denotes the transverse asymptotic region, which can be viewed as asymptotic $\mathbb{R}\times \mathbb{R}^{d-1}$ or $\mathbb{R}\times \mathbb{H}^{d-1}$, with the first factor representing the bulk radial direction.}
	\label{Asymptotic}
\end{figure}

This assumption is essential for making the infinite-volume averages in Eq.~\eqref{F} well-defined. Without it, the integrated curvature and extrinsic curvature terms would grow with the cutoff $L$ in an uncontrolled manner, and the limit  $\langle X\rangle_z:=\lim_{L\to\infty}\mathcal A^{-1}(L)\int_{\mathfrak B_{z,L}}X\,\mathrm dS$ could diverge or be path-dependent. Moreover, the key inequality $\langle \mathfrak R\rangle_z\leq 0$ used in Eq.~\eqref{eq:Fprimegauge} relies crucially on this assumption: any positive-curvature region is confined to a compact set, so its contribution to the average vanishes as the volume diverges, leaving only the nonpositive asymptotic curvature (zero for planar, negative for hyperbolic) to survive the averaging. Though this is not a weak assumption in mathematical sense, it precisely covers the physically relevant scenarios in holographic duality, where planar or hyperbolic AdS black holes serve as duals of quantum matter, and inhomogeneities (e.g., localized impurities or modulated chemical potentials) naturally appear as compactly supported geometric deviations. We leave more general case to future study. 


\subsection{Definition of the flux function $F(z)$}

We foliate the stationary pocket $\mathcal{D}$ by a family of $(d-1)$-dimensional leaves $\mathcal{S}_{z}$. There exists a smooth function $z:\mathcal{D}\rightarrow [z_1,z_2]$ such that $\mathcal{D}=\bigcup_{z\in[z_1,z_2]}\mathcal{S}_z$, where $\mathcal{S}_z:=\{p\in\mathcal{D}\mid z(p)=\mathrm{const.}\}$ is a level set of $z$, and the boundary satisfies $\mathcal{S}_{z_1}=H_1$ and $\mathcal{S}_{z_2}=H_2$ (Note that $\mathcal H$ denotes the $d$-dimensional horizon, while $H=\mathcal H\cap\Sigma_t$ denotes its $(d-1)$-dimensional horizon section on $\Sigma_t$). We require that the foliation preserves the symmetry generated by $\xi^\mu_{(\phi)}$, so that $\xi^\mu_{(\phi)}$ is tangent to every leaf $\mathcal{S}_z$. 

Under the asymptotic assumption in Sec.~S2.1, the leaves $\mathcal{S}_{z}$ inherit asymptotically planar or hyperbolic geometry along the noncompact directions. More precisely, outside a sufficiently large compact region $\mathfrak{B}_{z,L}$, each $\mathcal{S}_z\setminus \mathfrak{B}_{z,L}$ is diffeomorphic to an end of either $\mathbb{R}^{d-1}$ or $\mathbb{H}^{d-1}$, and its induced metric approaches the corresponding planar or hyperbolic metric at infinity.

Let $s^a$ be the unit normal to $\mathcal{S}_{z}$ within $\mathcal{D}$, oriented such that $s^a D_{a}z>0$, and let $\varphi$ be the lapse of the foliation, defined by
\begin{equation}
    s^aD_a z=\varphi^{-1}\,.
\end{equation}
The normal evolution of the foliation is therefore generated by $\varphi s^a$, which advances the foliation parameter $z$ at unit rate.

\begin{figure}[tbph]
	\centering
	\includegraphics[width=7cm]{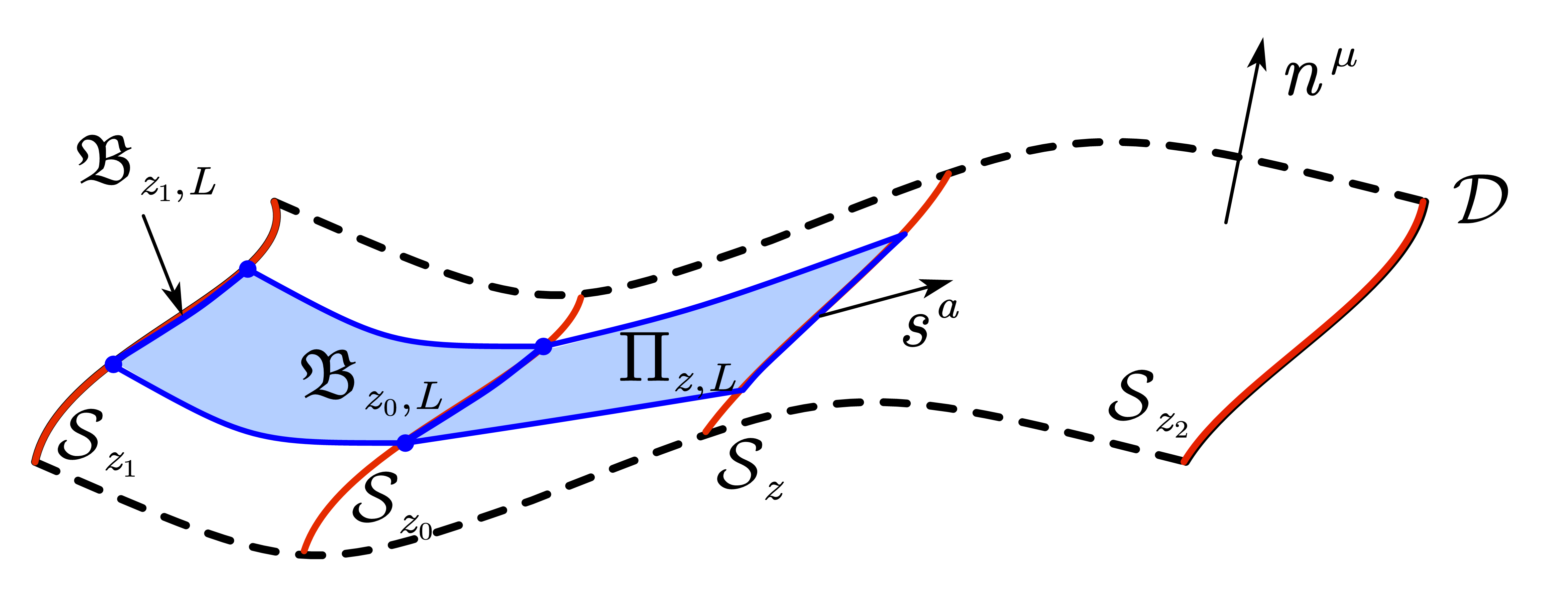}
    \captionsetup{
    justification=raggedright,
     singlelinecheck=false
    }
	\caption{The section of stationary pocket $\mathcal{D}$ is foliated by a series of $(d-1)$-dimensional level sets $\mathcal{S}_{z}$, which is denoted by red lines. The shadow region is $\Pi_{z,L}:=\bigcup_{z_k\in[z_1,z]}\mathfrak{B}_{z_{k}}$.}
	\label{S2.1A}
\end{figure}

To define the flux $F(z)$ via integrals over the noncompact leaves, we first choose on the reference leaf $\mathcal{S}_{z_{1}}$ (i.e., the level set $z=z_1$) a compact cut-off region $\mathfrak{B}_{z_{1},L}$ (as shown in Fig.~\ref{S2.1A}), whose boundaries lie in the asymptotic planar or hyperbolic region of $\mathcal{S}_{z_{1}}$, such that the regions exhaust $\mathcal{S}_{z_1}$ as $L\rightarrow\infty$. Once $\mathfrak{B}_{z_1,L}$ is fixed, its boundary is propagated along the
normal evolution of the prescribed foliation, and the compact region enclosed by the propagated boundary on $\mathcal{S}_z$ is denoted by $\mathfrak{B}_{z,L}$. Thus, for a fixed $L$, the cutoff regions $\mathfrak{B}_{z,L}$ on different leaves correspond to the same reference cutoff $\mathfrak{B}_{z_1,L}$ transported through the foliation.

Let $\gamma_{IJ}$ be the induced metric on $\mathcal{S}_z$, and denote by
$\mathfrak{D}_A$ its covariant derivative and by $\mathfrak{R}$ its
scalar curvature. Let $\mathcal{K}_{IJ}$ be the extrinsic curvature of
$\mathcal{S}_z\subset\Sigma_t$, with
$\mathcal{K}=\gamma^{IJ}\mathcal{K}_{IJ}$. The flux function in the main text reads
\begin{equation}
F(z):=\lim_{L\rightarrow\infty}\frac{1}{\mathcal{A}(L)}
\int_{\mathfrak B_{z,L}}
\left[s^a\partial_aN-\frac{\mathcal K N}{d-1}\right]\mathrm dS,
\label{FAd0}
\end{equation}
where $\mathcal{A}(L):=\int_{\mathfrak B_{z_1,L}}\mathrm dS$. For convenience, we omit the symbol $\lim_{L\to\infty}$ and always assume $L\to\infty$ in the following.

\subsection{Derivation of $F'(z)$}
As shown in Fig.~\ref{S2.1A}, we define the domain $\Pi_{z,L}:=\bigcup_{z_k\in[z_1,z]}\mathfrak{B}_{z_{k}}$. Then from Eq.~\eqref{FAd0} we obtain
\begin{align}
{F}'(z)&=\frac{{F}(z+\Delta z)-{F}(z)}{\Delta z}\notag\\
&=\frac{1}{\mathcal{A}\Delta z}\int_{\partial(\Delta\Pi_{z,L})}\left[ \partial_{a}N-\frac{\mathcal{K}Ns_a}{d-1} \right]\text{d}S^a\,.\label{SA}
\end{align}
Here region $\Delta\Pi_{z,L}=\Pi_{z+\Delta z,L}\backslash\Pi_{z,L}$. Note $\partial(\Delta\Pi_{z,L})$ does also contain the boundary $E=\bigcup_{z_{k}\in[z,z+\Delta z]}\partial\mathfrak{B}_{z_{k},L}$ (as shown in Fig.~\ref{S2.2A}). 
However, these lateral boundaries do not contribute in the limit $L\rightarrow\infty.$ Let $\nu^{a}$ denote the outward unit normal to the lateral boundary. Since $\nu^{a}$ is tangent to each leaf, it satisfies $\nu^{a}s_{a}=0$. For the asymptotically planar or hyperbolic ends considered here, the lapse $N$ becomes asymptotically constant on each leaf: $N \xrightarrow{L\rightarrow\infty}N_{\infty}(z)$, and hence its tangential gradient $\nu^{a}\partial_{a}N$ vanishes in the limit $L\rightarrow \infty$. Thus the contribution from the boundary $E$ vanishes and so Eq.~\eqref{SA} holds. 
\begin{figure}[tbph]
	\includegraphics[width=4.5cm]{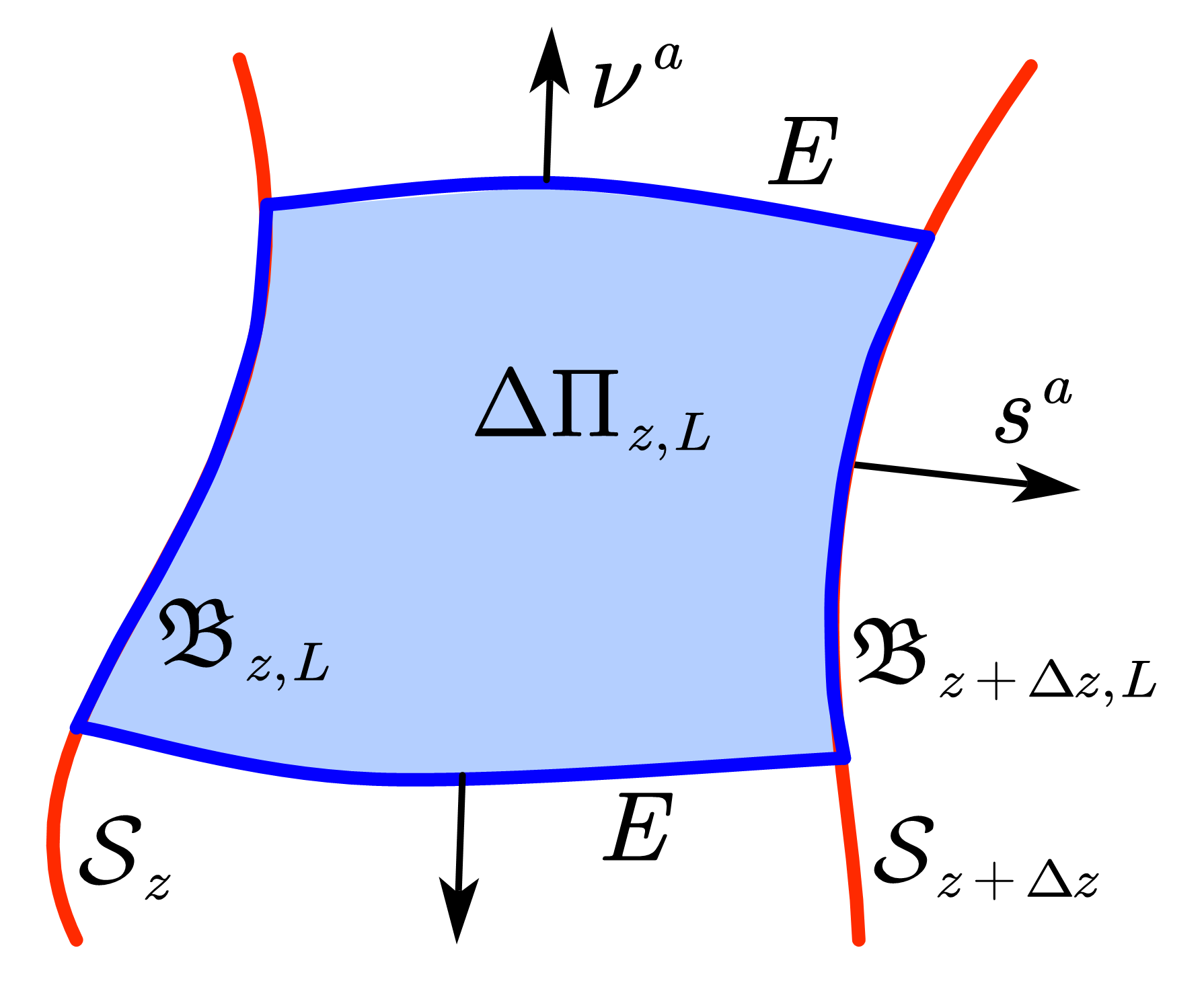}
    \captionsetup{
    justification=raggedright,
     singlelinecheck=false
    }
	\caption{Schematic illustration of $\Delta \Pi_{z,L}$ between $\mathfrak{B}_{z,L}\subset\mathcal{S}_{z}$ and $\mathfrak{B}_{z+\Delta z,L}\subset\mathcal{S}_{z+\Delta z}$. The lateral boundary $E$ is generated by propagating the cutoff boundary along the foliation. Its outward unit normal vector is $\nu^a$, which is tangent to each leaf and therefore satisfies $\nu^{a}s_{a}=0$.}
	\label{S2.2A}
\end{figure}

Applying Stokes' theorem gives
\begin{align}
    &\int_{\partial(\Delta\Pi_{z,L})}\left[ \partial_{a}N-\frac{\mathcal{K}Ns_a}{d-1} \right]\text{d}S^a\notag\\=&\int_{\Delta\Pi_{z,L}}\left[ D^2N-\frac{D_{a}(\mathcal{K}Ns^a)}{d-1} \right]\text{d}V,
\end{align}
where $\text{d}V$ denotes the volume element of $\Delta\Pi_{z,L}$. 
With the foliation lapse $\varphi$, we have $\text{d}V=\varphi\,\text{d}z\,\text{d}S$. Furthermore,
\begin{align}
D^{a}(\mathcal{K}Ns_{a})=Ns^{a}\partial_{a}\mathcal{K}+N\mathcal{K}^2+\mathcal{K}s^{a}\partial_{a}N, \label{q}
\end{align}
together with Eq.~(6) and Eq.~(C.15) of Ref.\cite{10.1093/ptep/ptx072}
\begin{align}
s^a \partial_a \mathcal{K} = -\frac{1}{2} \left( {}^{(d)}R - \mathfrak{R} + \mathcal{K}^2 + \mathcal{K}_{IJ} \mathcal{K}^{IJ} \right) - \varphi^{-1} \mathfrak{D}^2 \varphi, \label{w}
\end{align}
\begin{align}
\mathcal{K}_{IJ} \mathcal{K}^{IJ}  =& \mathcal{K}^2 - \mathfrak{R} - 2 G_{ab} s^a s^b + \frac{2}{N} \mathfrak{D}^2 N \notag \\
&+\frac{2\mathcal{K}}{N} s^a \partial_a N+\frac{\Psi^2}{2N^2}\bigl[(s^a \partial_{a}\omega)^2-(\mathfrak{D}\omega)^2\bigr]. \label{e}
\end{align}
Here $G_{ab}=G_{\mu\nu}{e^{\mu}}_{a}{e^{\nu}}_{b}$ is the projection of the Einstein tensor onto $\Sigma_t$ and ${^{(d)}}R$ is scalar curvature of $d$-dimensional spacelike hypersurface $\Sigma_t$. Combining Eqs.~\eqref{q}, \eqref{w} and \eqref{e} we obtain
\begin{align}
{F}'(z)
=& \frac{1}{(d-1)\mathcal{A}} \int_{\mathcal S_{z,L}} \left\{ (d-1)D^2N+ {^{(d)}R}N/2 \right. \notag \\  &\left. - \mathfrak{R} N- NG_{ab}s^a s^b  + \varphi^{-1}N\mathfrak{D}^2\varphi + \mathfrak{D}^2N \right.\notag \\&\left.+
\frac{\Psi^2}{4N}\left[(s^a\partial_{a}\omega)^2-(\mathfrak{D}\omega)^2\right]\right\} \varphi \text{d}S.
\end{align}
Define $\rho:=T_{\mu\nu}n^{\mu}n^{\nu}$ as the energy density, and $\mathcal{T}_{ab}:=T_{\mu\nu}{e^{\mu}}_{a}{e^{\nu}}_{b}$ as the projection of the stress-energy tensor onto $\Sigma_t$. Using the Einstein constraint equations in ADM decomposition
\begin{align}
^{(d)}R+K^2-K_{ab}K^{ab}=2\rho, \label{const1}
\end{align}
\begin{align}
D^2N=&N\left({}^{(d)}R+K^2\right)-N\left[ \mathcal{T}+\frac{d}{d-1}(\rho-\mathcal{T}) \right], \label{const2}
\end{align}
we can eliminate $D^2N$ and ${}^{(d)}R$. Noting that the extrinsic curvature $K_{ab}$ of $\Sigma_{t}$ reads
\begin{align}
K_{ab}= \frac{1}{2N}(\Psi_a D_b\omega+\Psi_b D_a\omega),
\end{align}
we have
\begin{align}
K^{ab}K_{ab}=\frac{\Psi^2}{2N^2}(D\omega)^2=\frac{\Psi^2}{2N^2}\left[(s^a\partial_{a}\omega)^2+(\mathfrak{D}\omega)^2\right].
\end{align}
Finally we arrive at
    \begin{align}\label{dfdz1}
{F}'(z)
=&\lim_{L \to \infty} \frac{1}{(d-1)\mathcal{A}(L)}\int_{\mathfrak{B}_{z,L}} \Bigg[ N\varphi(\varpi - \mathfrak{R}) \notag \\&  + \mathfrak{D}^2(N\varphi) - 2\gamma^{IJ}(\partial_I\varphi)(\partial_J N) \notag\\
& +\frac{\Psi^2\varphi}{N}\left[\frac{d}{2}(s^a \partial_{a}\omega)^2+\frac{d-1}{2}(\mathfrak{D}\omega)^2\right]\Bigg] \text{d}S .
\end{align}
The quantity $\varpi$ reads
\begin{align}
    \varpi&=\rho d+T-\mathcal{T}_{ab}s^{a}s^{b}\notag \\
    &=\frac{1}{2}\sum_{I=1}^{d-1}(T_{\mu\nu}\ell_{I,+}^\mu\ell_{I,+}^\nu+T_{\mu\nu}\ell_{I,-}^\mu\ell_{I,-}^\nu),
\end{align}
where $\ell_{I,\pm}^\mu=n^\mu\pm e_I^\mu$ are null vectors. Here $\{n^\mu,s^\mu,e_I^\mu\}$ is a local orthonormal frame, with $e_I^\mu$ tangent to $\mathcal S_z$ ($I=1,2,\dots,d-1$). The Eq.~\eqref{dfdz1} gives Eq.~\eqref{eq:fluxidentity} in the main text. 

\subsection{Sign of $\langle \mathfrak{R}\rangle_z$ }\label{S21}
For the case $d<3$, we have $\mathfrak{R}=0$ and so $\langle \mathfrak{R}\rangle_z=0$. We now focus on $d\geq3$. For convenience, in the following we adopt the notation in the main text, writing $\langle X\rangle_{z}:=\lim_{L\rightarrow \infty }\mathcal{A}^{-1}(L)\int_{\mathfrak{B}_{z,L}}$ as the average density of $X$ on $\mathfrak{B}_{z,L}$. Thus, the averaged curvature density $\langle \mathfrak{R}\rangle_z$ reads
\begin{equation}
  \begin{aligned}
    \langle \mathfrak{R} \rangle_z=\lim_{L\rightarrow\infty}\frac{1}{\mathcal{A}(L)}
    \int_{\mathfrak B_{z,L}}\mathfrak{R}\mathrm dS ,
  \end{aligned}
\end{equation}
where $\mathcal{A}(L):=\int_{\mathfrak B_{z_1,L}}\mathrm dS$. It can be seen that $\langle \mathfrak{R}\rangle_z$ is a $z$ function but $\mathcal{A}(L)$ is not.

In our setup, we assume that the transverse asymptotic region of $\Sigma_t$ possesses planar or hyperbolic symmetry and that every section $\mathcal{S}_z$ approaches the corresponding symmetric geometry at transverse infinity. We further assume that all nonuniform contributions to the curvature, including any region of positive curvature, are confined to a compact subset $\mathfrak{B}_{z,L'}\subset\mathcal{S}_z$. Because the curvature integral over this compact region remains finite while the volume of the averaging domain diverges, its contribution to the averaged curvature density vanishes:
\begin{equation}
\lim_{L\to\infty}
\frac{1}{\mathcal{A}(L)}
\int_{\mathfrak{B}_{z,L'}}
\mathfrak{R}\,\mathrm{d}S
=
0.
\end{equation}
Consequently, the averaged curvature density is determined entirely by the noncompact asymptotic region. For planar asymptotics, this contribution vanishes, whereas for hyperbolic asymptotics it is negative. Therefore, the averaged curvature density on each section is nonpositive:
\begin{equation}
\langle \mathfrak{R}\rangle_z\leq 0.
\end{equation}

For the case $d=3$, we can obtain same result by a weak requirement: we only need to assume that transverse asymptotic region of horizon $\mathcal{H}_1$ possesses planar or hyperbolic symmetry. 

Following above argument of general $d$, we first find that the horizon satisfies $\langle\mathfrak{R}\rangle_{z_1}\leq 0$. To determine the sign of $\langle\mathfrak{R}\rangle_{z}$ for other $z$, we first use following formulas
\begin{equation}
    \begin{aligned}
        \delta\int_{\mathfrak B_{z,L}}\mathfrak{R}\mathrm{d}S&=\int_{\mathfrak B_{z,L}}G^{IJ}\delta \gamma_{IJ}\mathrm{d}S,\\
        \delta \gamma_{IJ}&=2\varphi \mathcal K_{IJ}\delta z.
    \end{aligned}
\end{equation}
The differentiation of $\langle \mathfrak{R}\rangle_z$ can be written as
\begin{equation}
\frac{\mathrm{d}}{\mathrm{d}z}
\langle \mathfrak{R}\rangle_z=2\left\langle \varphi G^{IJ}\mathcal{K}_{IJ}\right\rangle_z\, ,
\end{equation}
where $N$ is the lapse function appearing in the metric ansatz~\eqref{metricAd}, $G^{IJ}$ is the Einstein tensor of $\mathcal{S}_z$, and $\mathcal{K}_{IJ}$ and $\mathcal{K}$ denote, respectively, the extrinsic curvature of $\mathcal{S}_z$ and its trace. In the case $d=3$, each section $\mathcal{S}_z$ is a two-dimensional Riemannian surface. The Einstein tensor of any two-dimensional Riemannian manifold vanishes identically, namely,
\begin{equation}
G^{IJ}=0.
\end{equation}
It therefore follows that
\begin{equation}
\frac{\mathrm{d}}{\mathrm{d}z}\langle \mathfrak{R}\rangle_z=0, \quad \langle \mathfrak{R}\rangle_z
=
\langle \mathfrak{R}\rangle_{z_1}.
\label{dR2}
\end{equation}
The sign of the averaged curvature density is preserved along the flow and so we have $\langle \mathfrak{R}\rangle_z=\langle \mathfrak{R}\rangle_{z_1}\leq 0$.

\section{Weak Continuation of the Inverse Lapse Flow}\label{S3}
In the main text, the monotonicity argument for noncompact horizon sections case is formulated using a family of leaves $\mathcal S_z$ satisfying the inverse lapse flow $\varphi N=1$. If this family remained a globally smooth foliation connecting the two inner horizons, the nonnegativity of $F'(z)$ would immediately imply $F(z_2)\geq F(z_1)$, contradicting to $F(z_2)<F(z_1)$ required by existence of the stationary pocket. A global smooth foliation, however, is not generally guaranteed. Four distinct obstructions may arise. 
\begin{itemize}
\item \textit{Divergency of speed at the initial surface.} It requires that the flow starts from one horizon $z=z_1$. However, the lapse function $N=0$ at $z=z_1$, which leads to the flow speed $\varphi=1/N\rightarrow\infty$.
\item \textit{Focal-cusp formation.} The evolving leaf may develop a cusp when neighboring normal trajectories focus.

\item \textit{Self-intersection.} As a leaf propagates under the inverse lapse flow, different portions of a same evolving leaf may meet, causing the leaf to become self-intersecting.

\item \textit{Two-boundary overdetermination.} A flow that starts from one horizon does not reach the other horizon simultaneously, leading to a global two-boundary compatibility problem.

\end{itemize} 
For these four situations, we introduce a weak continuation of the inverse lapse flow based on the viscosity solution of the eikonal equation. We then show that, in this weak continuation, the inequality $F(z_{2})\ge F(z_{1})$ established in the main text remains valid.

\subsection{Inverse lapse flow}


To study the weak continuation of the inverse lapse flow, we start from the foliation setup and notation introduced in Sec.~\ref{S2}. Since $D_a z$ is normal to the level set $\mathcal S_z$, one has $D_a z=\varphi^{-1}s_a$. The inverse lapse gauge $\varphi N=1$ adopted in the main text is therefore equivalent to the eikonal equation
\begin{align}\label{ekeq1}
    h^{ab}D_a zD_b z=N^2.
\end{align}
Here $h_{ab}$ is the induced metric on $\Sigma_{t}$. This foliation is closely analogous to the inverse mean curvature flow (IMCF)~\cite{Jang:1977kef,https://doi.org/10.1111/j.1749-6632.1973.tb41445.x}, in which the evolving leaf moves in its normal direction with speed given by the inverse of its mean curvature $\mathcal{K}$ (the trace of the extrinsic curvature of $\mathcal{S}_{z}$). In the present case, the mean curvature $\mathcal{K}$ is replaced by the ADM lapse $N$: the normal speed with respect to the flow parameter $z$ is $N^{-1}$. The flow can be visualized as follows. As shown in Fig.~\ref{S3.1A}, start from a smooth initial leaf $S_{z_0}$, and let $s^a$ be its unit normal. For an infinitesimal increase $\text{d}z$, each point of $S_{z_0}$ is moved forward along its normal $s^a$ by the proper distance
\begin{align}
    \text{d}\ell=\varphi\,\text{d}z=\frac{\text{d}z}{N}.
    \label{dell}
\end{align}
All these points then form the next leaf $S_{z_0+\text{d}z}$. Taking this new leaf as the initial surface and repeating the same normal displacement yields $S_{z_0+2\text{d}z}$, and repeating this procedure generates a series of leaves $\{\mathcal S_z\mid z\geq z_0\}$. Thus the gauge $\varphi N=1$ defines a geometric flow: each leaf is generated from the preceding one by moving along the normal with speed $1/N$. This procedure is valid as long as the evolving leaves remain smooth and the normal trajectories do not intersect.
\begin{figure}[tbph]
	\centering
	\includegraphics[width=7cm]{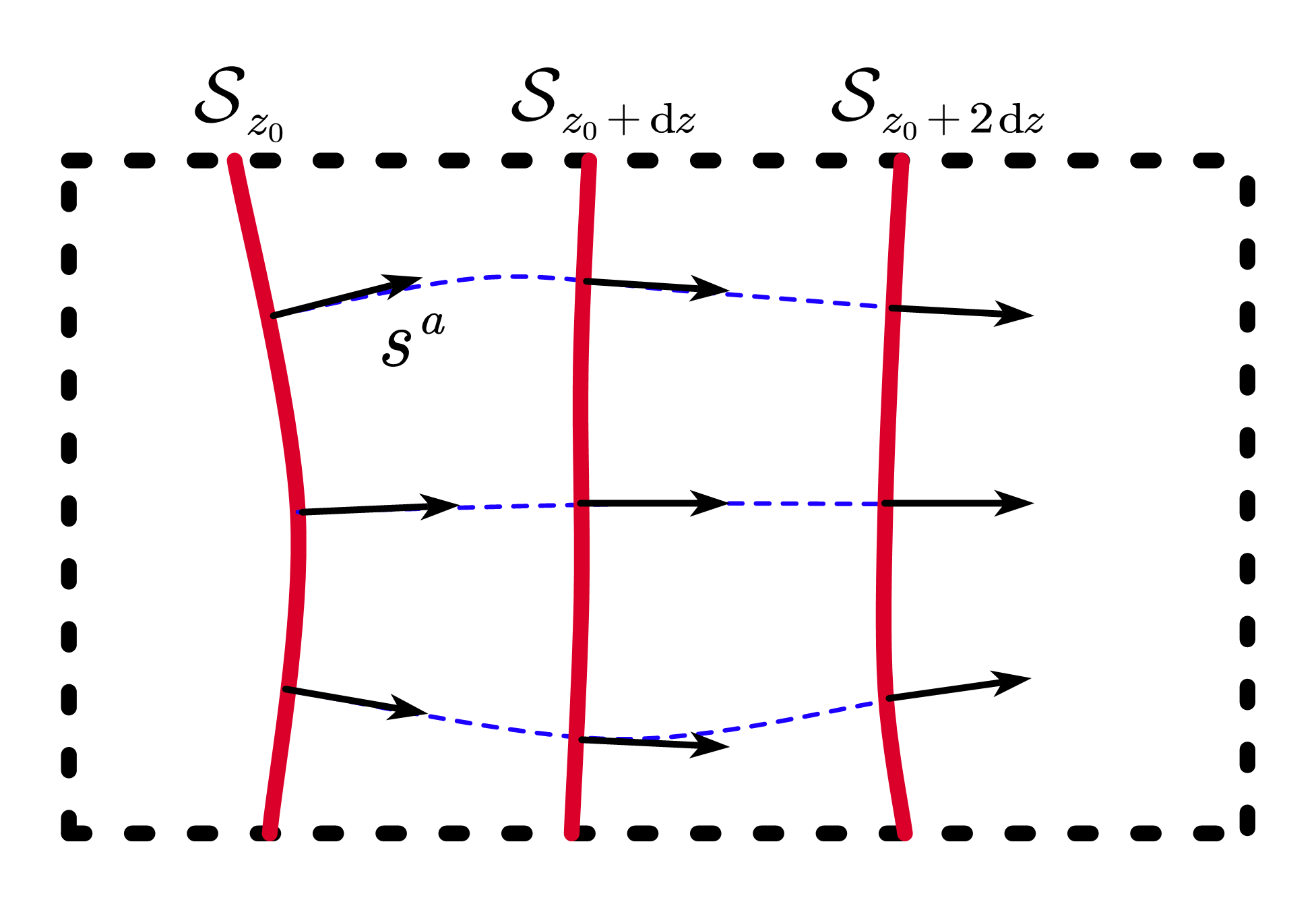}
    \captionsetup{
    justification=raggedright,
     singlelinecheck=false
    }
	\caption{The schematic diagram about inverse lapse flow. The flow is generated by the vector field of which the direction is along the normal vector $s^a$ and speed is the inverse of lapse function $N$.}
	\label{S3.1A}
\end{figure}


\subsection{Why smooth inverse lapse flow may fail}
Now let us explain in detail how the four distinct obstructions arise and why the classical smooth inverse lapse flow may fail. 

\paragraph*{Case A: Divergency of speed at initial surface.} The flow is originally defined by the gauge $\varphi N=1$ and so the local speed of the flow is $\varphi=1/N$. To compare the values of $F(z)$ at two horizons, the flow should start from one horizon. However, at the horizon we have $N=0$, and so the flow loses validity at the initial surface. 

\paragraph*{Case B: Focal-cusp formation.}
For a given lapse $N$ and the initial leaf $\mathcal{S}_{z_{0}}$, the normal trajectories star form neighboring points may gradually converge as the inverse flow proceeds. When this convergence occurs, the normal trajectories focus at a common point at some finite parameter value $z$ (as shown by the red solid curves in Fig.~\ref{S3.2A}). This focal point corresponds to a cusp on the evolving leaf, where the two sides meet with distinct limiting normals. Thus, the leaf is no longer smooth, and the unit normal $s^a$ and mean curvature $\mathcal{K}$ are no longer defined in the usual sense at the focal-cusp. We should therefore adopt a regularized integral formulation to handle the curvature quantities at the singular set.

\begin{figure}[tbph]
	\centering
	\includegraphics[width=6cm]{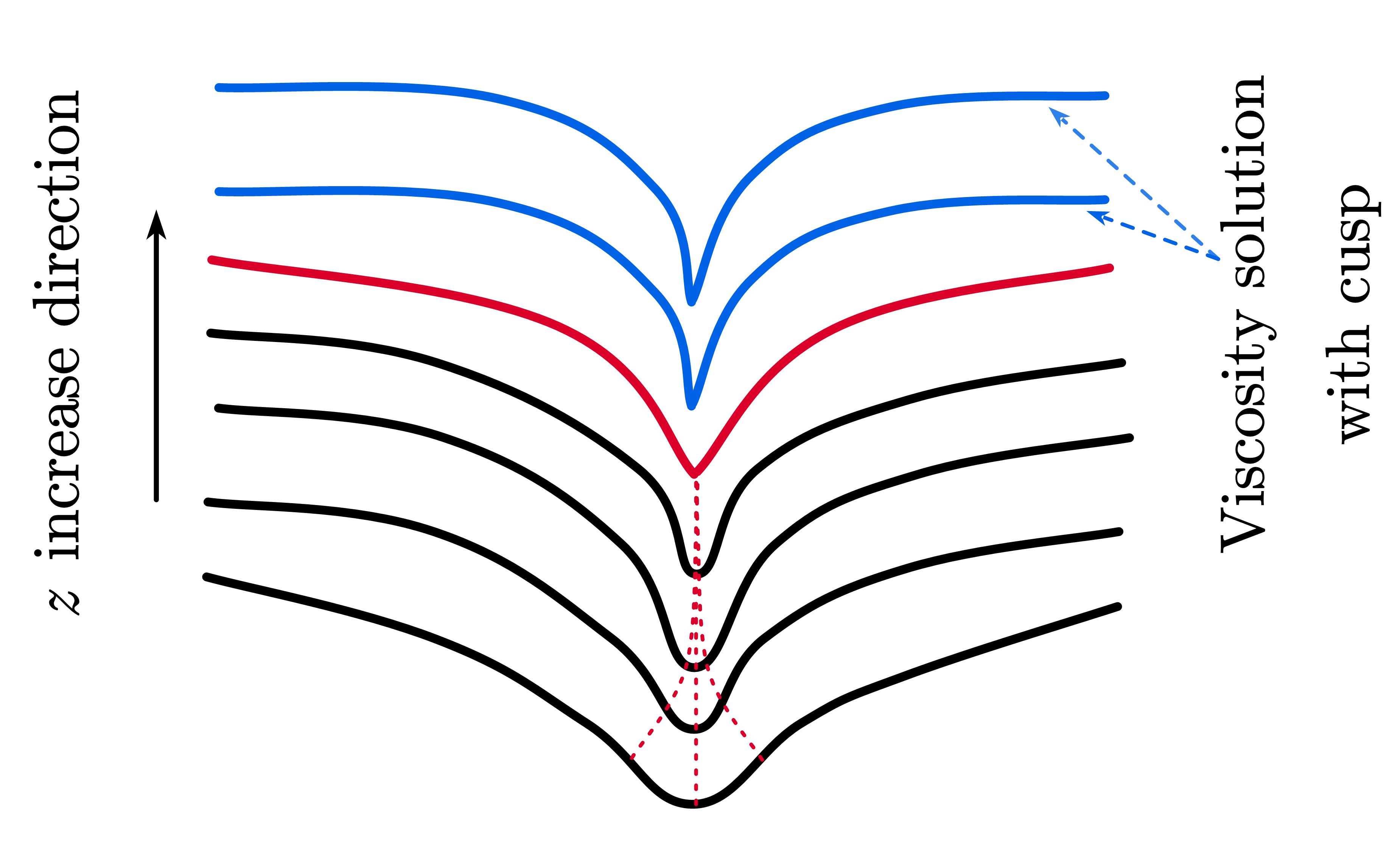}
    \captionsetup{
    justification=raggedright,
     singlelinecheck=false
    }
    	\caption{Schematic diagram of focal-cusp formation. The black solid curves represent the evolving smooth leaves of the foliation. The red dashed curves are normal trajectories, which progressively converge toward the caustic point. The red solid curve is the leaf containing the caustic point, where the normal trajectories meet and the leaf develops a cusp. The blue solid curve represents the weak solution (viscosity solution), which contains the cusp in the leaf.}
	\label{S3.2A}
\end{figure}

\paragraph*{Case C: Self-intersection.}
Due to the possible inhomogeneity of the lapse function $N$, the normal evolution speed $1/N$ varies across different regions of a leaf. This speed disparity causes two initially separated points on the same evolving leaf, following their respective normal trajectories, to eventually meet at the same point of $\mathcal D$, causing the leaf to intersect itself. Before this happens, the regions around these two points evolve independently and remain smooth and well separated. At the moment of contact, the leaf meets itself, forming a self-intersection point $\mathcal P_z$ (see Fig.~\ref{S3.2B}). At the self-intersection, the normal to the level set is no longer unique. Consequently, the unit normal $s^a$ and the mean curvature $\mathcal K$ cannot be defined as single-valued functions at that point.

This situation differs from focal-cusp formation. In the focal-cusp case, neighboring normal trajectories gradually converge and focus, with the separation between neighboring points being compressed to zero. The self-intersection is not caused by the local focusing. Rather, it arises through a different mechanism: the two distinct points of the same leaf are transported along different normal trajectories and arrive at the same point.

\begin{figure}[tbph]
	\centering
	\includegraphics[width=6.5cm]{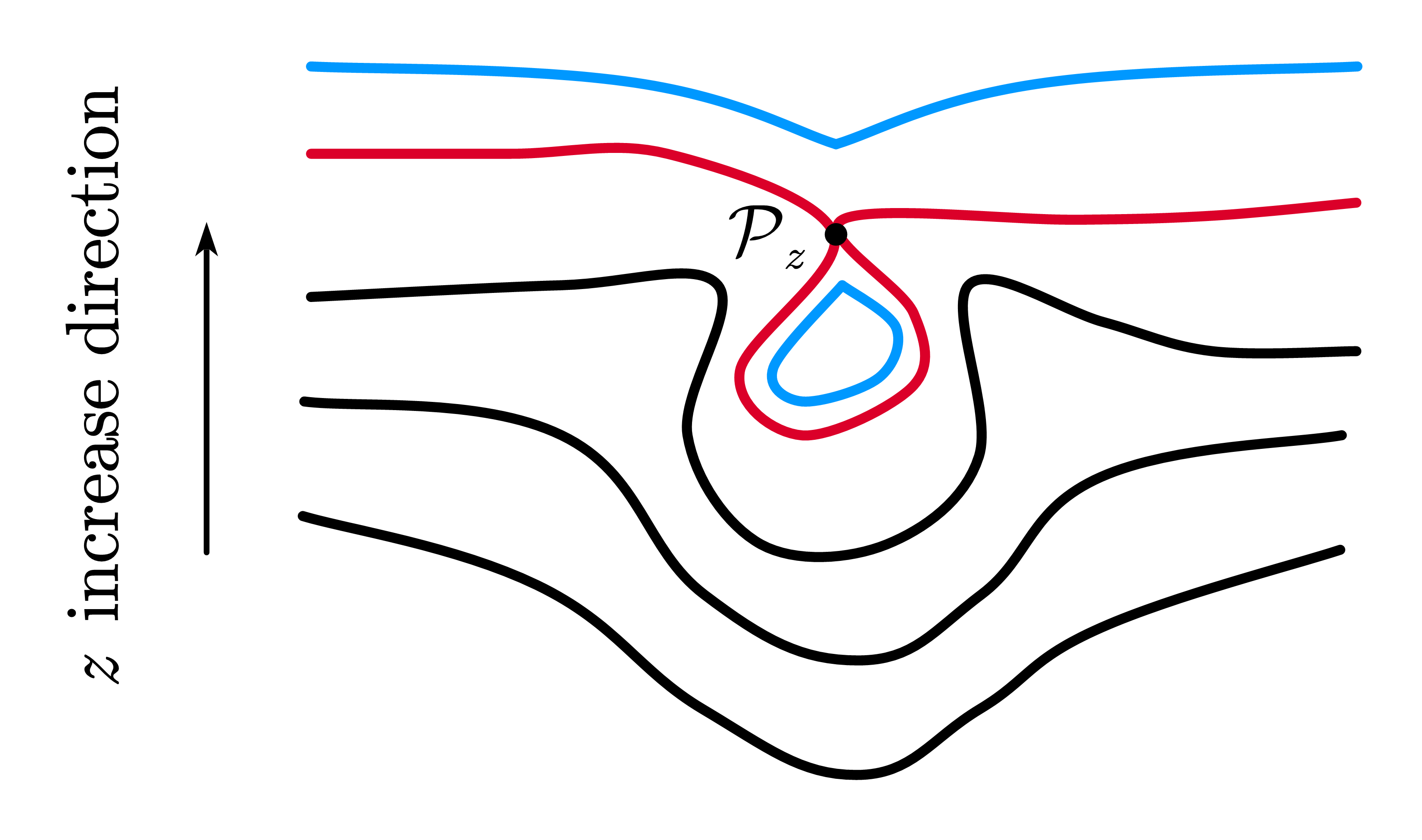}
    \captionsetup{
    justification=raggedright,
     singlelinecheck=false
    }
	\caption{Schematic diagram of self-intersection. The smooth leafs (black curves) develop a self-intersection singularity (red curve). After then smooth inverse lapse flow is broken. The leaf (blue curve) that follows will be given by a weak solution (viscosity solution), which contains two disconnected branches: one branch forms a cusped wavefront that continues to propagate, while the other forms a closed cusped component that gradually contracts.}
	\label{S3.2B}
\end{figure}

\paragraph*{Case D: Two-boundary overdetermination.}
This obstruction is not a local singularity of an evolving leaf, but a global compatibility problem. The region $\mathcal D$ discussed in the main text is bounded by two inner Killing horizon sections, denoted respectively by $H_1$ and $H_2$. And both boundaries are level sets of $z$, i.e., $z|_{H_1}=z_1$ and $z|_{H_2}=z_2$. Thus, one needs to find a global solution to the eikonal equation Eq.~\eqref{ekeq1} that satisfies both boundary conditions simultaneously. However, since the eikonal equation is first order, its solution is naturally determined by prescribing the value of $z$ on only one boundary and propagating the corresponding level sets into the interior. Prescribing the values of $z$ independently on both $ H_1$ and $ H_2$ adds an extra condition: the one-sided evolution generated from one horizon must reach the other horizon simultaneously. Such a compatibility condition is not guaranteed in a general geometry, although it may arise naturally in highly symmetric configurations, such as those with planar, spherical, or hyperbolic symmetry.

Unlike a focal cusp or a self-intersection, this obstruction may occur even when every leaf generated remains smooth. The problem is that the one-sided classical flow does not produce a single leaf that covers the complete second boundary.

\begin{figure}[tbph]
	\centering
	\includegraphics[width=6cm]{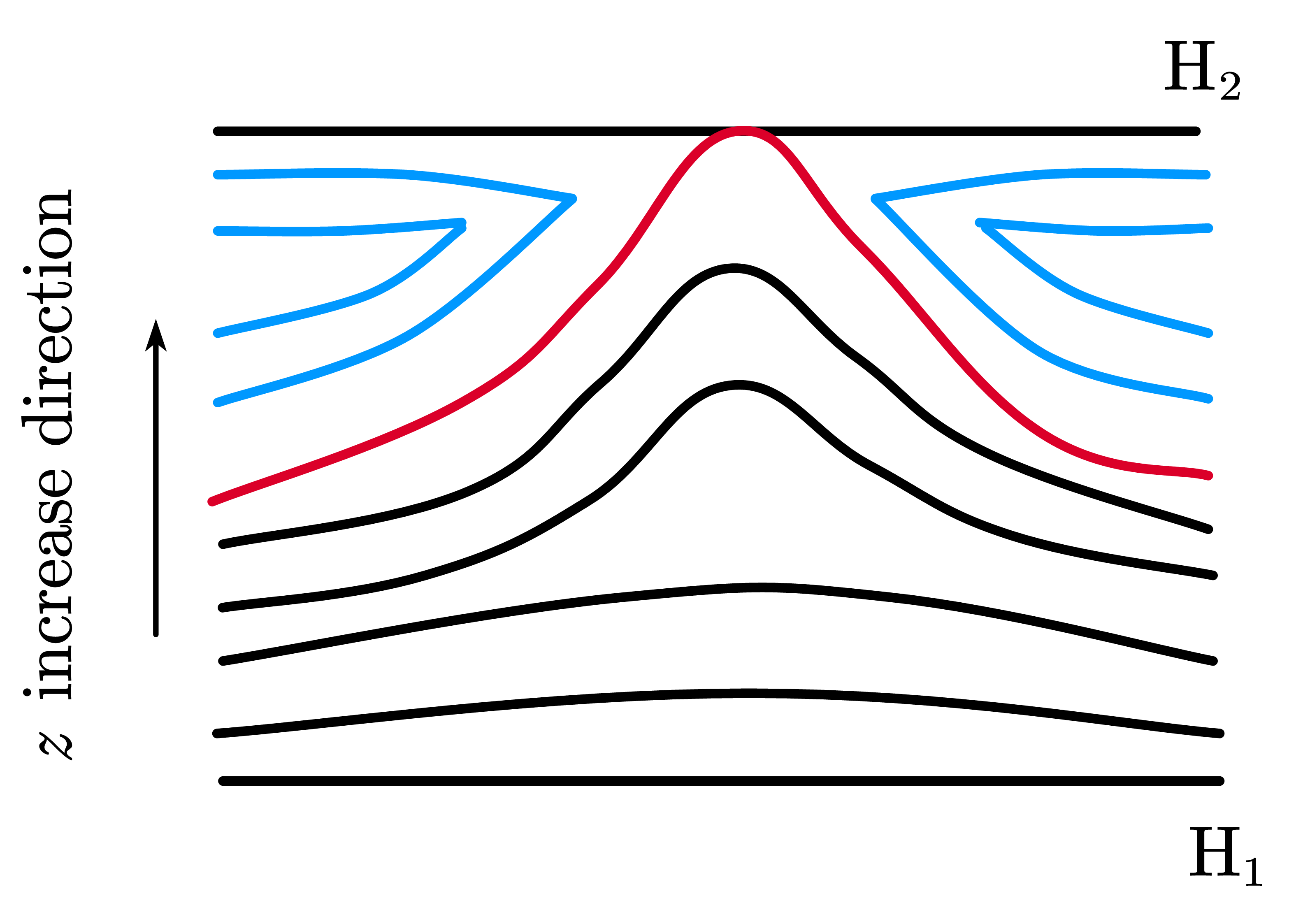}
    \captionsetup{
    justification=raggedright,
     singlelinecheck=false
    }
	\caption{Schematic diagram of two-boundary overdetermination. Black curves and the red curve are smooth leaves generated from $H_{1}$. Moreover, the red curves is the leaf $\mathcal{S}^{(-)}_{z_{c}}$ that first contact the second horizon $H_{2}$ at the value $z=z_{\text{c}}$. The blue curves denote leaves $\widetilde {\mathcal{S}}_{z}$ which are generated from the composite initial surface  $\widetilde {\mathcal{S}}_{z_{\text{c}}}=\mathcal{S}^{(-)}_{z_{c}}\cup H_{2}$}.
	\label{S3.2C}
\end{figure}

Let us make a short summary. \textit{Case A} describes divergency of flow speed at initial horizon. The eikonal equation~\eqref{ekeq1} and our inverse lapse flow cannot be defined directly at the initial horizon. \textit{Case B} and \textit{Case C} represent local singularities of the evolving leaves, while \textit{Case D} represents a global failure of boundary compatibility. In each case, the classical family $\{\mathcal S_z\}$ fails to provide a globally smooth foliation connecting the two horizons section. In the following sections, we use the viscosity solution of the eikonal equation to define a weak solution of the inverse lapse flow. We will show that it can overcome above divergency and singularities as well as preserves flux monotonicity used in the main text.

\subsection{Viscosity solution of the eikonal equation}

Once a focal cusp or a self-intersection forms, the inverse lapse flow can no longer be a classical smooth foliation. Nevertheless, the eikonal equation still admits a natural weak continuation in the viscosity sense. A viscosity solution is defined by testing a nondifferentiable function with smooth functions that touch it locally from above or below. The upper and lower touching functions define viscosity sub- and supersolutions, which provide the two opposite inequalities required by the Hamilton–Jacobi equation and thereby select the admissible weak solution through the comparison principle. Equivalently, the viscosity solution can be obtained through the standard vanishing-viscosity procedure, in which a small diffusion term is introduced to smooth the equation and is then taken to zero (see Chapter 10 of Ref.~\cite{evans2022partial} for a detailed construction). For our purposes, it is convenient to characterize the viscosity solution through its associated value function, following the standard Hamilton-Jacobi formulation presented in Ref.~\cite{evans2022partial}.

Let $\mathcal S_{z_1}$ be a smooth initial leaf in the interior of $\mathcal D$. For each point $x\in\mathcal D$, let $d_N(x,\mathcal S_{z_1})$ denote the minimum weighted length among all spatial curves connecting the initial leaf $\mathcal S_{z_1}$ to $x$, where the weighted length of each curve is measured with the weighted line element $N\,\mathrm d\ell$ (see Eq.\eqref{dell}). Thus, $d_N(x,\mathcal S_{z_1})$ represents the shortest $N$-weighted distance from $\mathcal S_{z_1}$ to $x$. The viscosity solution of eikonal equation~\eqref{ekeq1} is given by the following value function
\begin{align}
    z(x)=z_1+d_N\!\left(x,\mathcal S_{z_1}\right).
    \label{eq:value_function}
\end{align}
Equivalently, it gives the minimum value of the flow parameter $z$ at which the inverse lapse flow reaches $x$. The leaf $\mathcal  S_z$ for $z>z_1$ is then defined as the set of points with the same weighted distance from $\mathcal S_{z_1}$, 
\begin{align}
    \mathcal S_z= \left\{x\in\mathcal D\,\middle|\,z(x)=z\right\}.
\end{align}
In other words, the $\mathcal S_{z_1}\cup\mathcal{S}_{z}$ gives the boundaries of points that satisfy $0<d_N(x,\mathcal S_{z_1})\leq z-z_1$. 
This viscosity solution agrees with the classical inverse lapse flow wherever the evolution remains smooth. Along a smooth normal trajectory, the inverse lapse condition gives $\text{d}z=N\,\text{d}\ell$. When the minimizing curve from $S_{z_0}$ to $x$ is unique and varies smoothly with $x$, it follows the normal direction of the level sets of $z$. In this region, the value function is smooth, satisfies the eikonal equation in the classical sense, and reproduces the evolving leaves generated by the inverse lapse flow.

Near the horizon, the foliation is given by the viscosity solution ~\eqref{eq:value_function} of the eikonal equation~\eqref{ekeq1}, rather than by integrating the divergent speed $1/N$ from the horizon. Though the original inverse lapse flow defined by $\varphi N=1$ loses validity, the viscosity solution ~\eqref{eq:value_function} is well-defined. The horizon $z=z_1$ serves as reference surface to define the $N$-weighted distance function $d_N(x,\mathcal S_{z_1})$. This overcomes the speed divergency at the initial surface.

The importance of the value-function description appears when different normal trajectories focus or meet. Even after the classical leaf loses smoothness, Eq.~\eqref{eq:value_function} continues to assign a unique value of $z$ to every point. The function $z(x)$ therefore remains continuous and single-valued, although it need not remain differentiable on the singular set. Geometrically, the viscosity solution does not remove the cusp or smooth out the self-intersection. Instead, it provides a unique weak continuation of the evolving level sets after the classical normal construction has broken down.

When a focal cusp forms on the leaf, although the normal $s^a$ are not defined at the cusp, the viscosity solution admits that the leaf continues to evolve as a cusped wavefront within the framework of viscosity solutions (see the blue curves of Fig.~\ref{S3.2A} as a schematic diagram). In the case of self-intersection, the viscosity solution constructs the evolution of the leaves after the intersection through weighted distance. Upon self-intersection, the leaf forms a noncompact cusped wavefront $\mathcal M_z$ and a closed cusped component $\mathcal L_z$ (see the blue curve in Fig.~\ref{S3.2B} as a schematic diagram). The noncompact wavefront continues to propagate in the remaining region of $\mathcal D$, while the compact set evolves inward and gradually contracts.

The viscosity solution preserves the level-set description, but the singularities still persist: at the cusp, the normal and curvature remain undefined. In the following sub-section, we adopt a regularization procedure to handle the singularities at the cusp, and by which we will also explain why the flux monotonicity used in the main text is preserved. 

\subsection{Regularization of singularities in the weak leaf}\label{regs1}

The viscosity solution provides a continuous and single-valued continuation of the flow after the classical inverse lapse flow develops focal singularities or self-intersection. It keeps the cusps caused by focal focusing and self-intersection during the flow rather than eliminating them. Our flux function $F(z)$ contains the integration of extrinsic curvature, which is ill-defined at cusps. We will show that these cusp singularities are mild in the sense that they can be treated by the regularization procedure and will not break the monotonicity of flux function $F(z)$.

Let $\mathcal P_z\subset \mathcal S_z$ denote the cusp set in a weak leaf $S_z$. To define the geometric quantities at the cusp, we introduce a family of smooth regularized leaves $\mathcal S_z^\varepsilon$. For each $\varepsilon>0$, the regularized leaf $\mathcal S_z^\varepsilon$ is obtained by smoothing out the cusp of $\mathcal S_z$ in an $\varepsilon$-neighborhood of the cusp set $\mathcal P_z$ (as shown in Fig.~\ref{S3.4A}), while agreeing with the smooth branches of $\mathcal S_z$ outside this neighborhood. As $\varepsilon \rightarrow 0$, the regularized leaf converges to $\mathcal S_z$ in the limit, while the regularized neighborhood shrinks to the cusp set $\mathcal{P}_{z}$.

\begin{figure}[tbph]
	\centering
	\includegraphics[width=4.5cm]{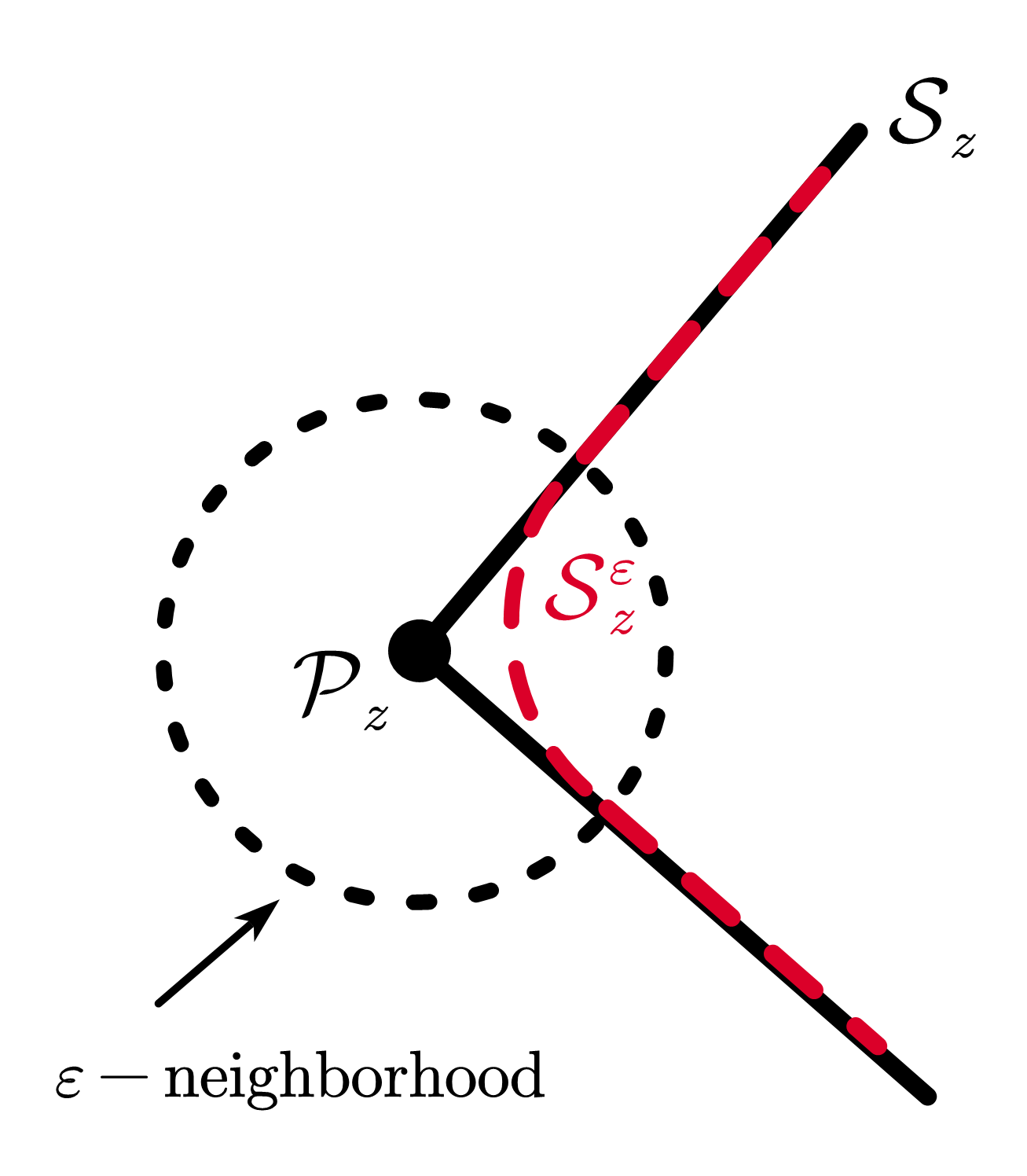}
    \captionsetup{
    justification=raggedright,
     singlelinecheck=false
    }
	\caption{Local regularization of a cusp in a weak leaf $\mathcal S_z$. The cusp $\mathcal P_z$ is replaced within an $\varepsilon$-neighborhood by a smooth segment $\mathcal S_z^\varepsilon$ (red dashed), which coincides with $\mathcal S_z$ outside the smoothing region. In the limit $\varepsilon\to0$, $\mathcal S_z^\varepsilon$ approaches $\mathcal S_z$.}
	\label{S3.4A}
\end{figure}

On the smoothing region, the curvature is well-defined on $\mathcal{S}_{z}^{\varepsilon}$. Its integral over $\varepsilon$-neighborhood, in the limit $\varepsilon \rightarrow 0$, is not zero but converges to the defect-angle contribution $\Theta$ supported on $\mathcal{P}_{z}$. Thus, the curvature integral over the weak leaf decomposes as
\begin{align}
    \int_{\mathcal{S}_{z}}\mathcal{K}\, \mathrm{d}S=\int_{\mathcal{S}_{z}\setminus \mathcal{P}_{z}}\mathcal{K} \, \mathrm{d}S +\int_{\mathcal{P}_{z}}\Theta\, \mathrm{d}\mu^{m}
\end{align}
where $\mathrm d\mu^{m}$ is the measure induced on the
$m$-dimensional cusp set ($m\le d-2$). Unlike the mean curvature term, the term $s^a\partial_aN$ remains bounded throughout the regularization, so it produces no contribution on $\mathcal{P}_{z}$. Thus, the flux $F(z)$ on weak leaf can be decomposed into a smooth part and a cusp part, 
\begin{equation}
    \begin{aligned}
    F(z)=\lim_{L\rightarrow\infty}
    \frac{1}{\mathcal A(L)}
    \Bigg[
    &
    \int_{\mathfrak B_{z,L}\setminus\mathcal P_z}
    \left(
        s^a\partial_a N
        -
        \frac{\mathcal K N}{d-1}
    \right)
    \mathrm dS
    \\
    &
    -
    \int_{\mathfrak B_{z,L}\cap\mathcal P_z}
    \frac{N\Theta}{d-1}
    \,
    \mathrm d\mu^{m}
    \Bigg].
    \end{aligned}
\end{equation}
The defect angle $\Theta$ and the lapse $N$ are bounded. Moreover, due to the assumed asymptotic planar or hyperbolic
behavior of the foliation, the defect angle disappears
asymptotically, with $\Theta\rightarrow0$ along the noncompact directions. Consequently, the defect angle term is finite, so that
\begin{align}\label{cusp}
    \lim_{L\rightarrow\infty}
    \int_{\mathfrak B_{z,L}\cap\mathcal P_z}
    \frac{N\Theta}{d-1}\,
    \mathrm d\mu^{m}<\infty.
\end{align}
Since the  denominator  $\mathcal A(L)\rightarrow \infty$ with $L \rightarrow \infty$, this defect angle term does not contribute to $F(z)$. 

\begin{figure}[tbph]
	\centering
	\includegraphics[width=8.5cm]{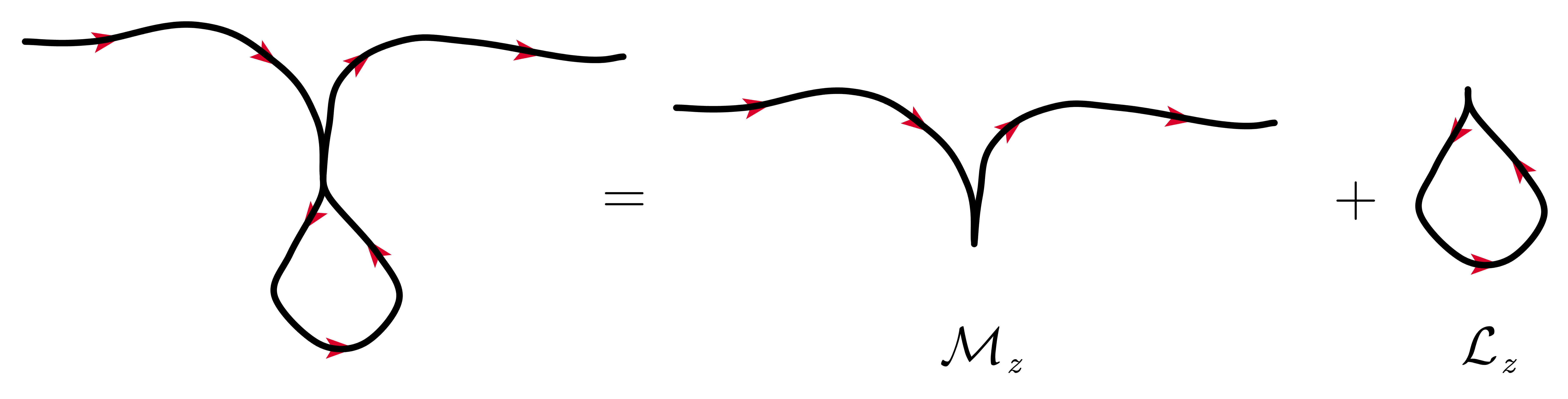}
    \captionsetup{
    justification=raggedright,
     singlelinecheck=false
    }
	\caption{Schematic illustration of the surgery at a self-intersection. The self-intersecting weak leaf decomposes into a noncompact wavefront $\mathcal M_z$ and a compact component $\mathcal L_z$ (the compactness is guaranteed by the asymptotic symmetry in transverse directions). We neglect the $\mathcal{L}_z$ part and the remainder part $\mathcal{M}_z$ then contains a cusp and can be handled by the same method as what we have done in Fig.~\ref{S3.4A}. }
	\label{Loop}
\end{figure}

For the self-intersection case, the weak leaf admits a particularly simple decomposition. As illustrated in Fig.~\ref{Loop}, when the self-intersection occurs, the leaf separates into a cusped noncompact wavefront $\mathcal{M}_{z}$, which continues toward transverse infinity, and a ``closed'' cusped component $\mathcal{L}_{z}$. Since we assume that deviation from planar or hyperbolic symmetry only occurs on a finite region (compact support), the component $\mathcal{L}_{z}$ is a compact surface. More generally, if several self-intersections occur on the same leaf, this leaf can be written as $\mathcal{S}_{z}=\mathcal{M}_{z}\cup\bigcup_{i}\mathcal{L}_{z}^{i}$. The flux $F(z)$ integral over this leaf therefore decomposes as
\begin{equation}
    \begin{aligned}
    F(z)
=
\lim_{L\rightarrow\infty}&
\frac{1}{\mathcal A(L)}
\Bigg[
\int_{\mathfrak B_{z,L}\cap\mathcal M_z}
\left(
s^a\partial_a N-\frac{\mathcal{K}N}{d-1}
\right)\mathrm{d}S
\\
&+
\sum_{i=1}^{n_z}
\int_{\mathcal L_z^{(i)}}
\left(
s^a\partial_a N-\frac{\mathcal{K}N}{d-1}
\right)\mathrm{d}S
\Bigg].
    \end{aligned}
\end{equation}
Here the cusp contributions are defined by the regularization procedure leading to Eq.~\ref{cusp}.

The asymptotic assumption plays an essential role. The weak leaves remain regular in the asymptotic planar or asymptotic hyperbolic region. Hence, a self-interaction cannot be generated at transverse infinity, and every closed component $\mathcal{L}_{z}^i$ is confined to the interior nonuniform region of the leaf. In particular, each $\mathcal{L}_{z}^i$ is contained in a compact subset of the leaf. After the cusp regularization procedure, its flux integral 
\begin{equation}
    \int_{\mathcal L_z^{(i)}}
\left(
s^a\partial_a N-\frac{\mathcal{K}N}{d-1}
\right)\mathrm{d}S<\infty.
\end{equation}
Assuming that the number $n_{z}$ of such closed components on each leaf is finite, their total contribution remains finite as $L\rightarrow\infty$. Since the denominator $\mathcal{A}(L)\rightarrow \infty$ with $L\rightarrow \infty$, each $\mathcal{L}_{z}^{i}$ does not contribute to $F(z)$.  

This observation allows us to perform a surgery on the weak leaf at each self-interaction: the closed components $\mathcal{L}_{z}^{i}$ are excised, while the noncompact wavefront $\mathcal M_{z}$ is retained as the new leaf for the subsequent inverse lapse evolution
%
%
%
and the flux function $F(z)$ is unchanged at this moment. In the future flow process, the leaf contains a cusp, and we can deal with it using the same method discussed in Fig.~\ref{S3.4A}.


Up to now, we have classified the possible singularities during the flow and given detailed methods on how to deal with these singularities. The monotonicity argument of the main text can therefore be extended to the weak leaves: for $z_a<z_b$, we have $F(z_b)\ge F(z_a)$. The flux $F(z)$ is a nondecreasing function of $z$.

\subsection{Two-boundary overdetermination}

We finally consider the global compatibility problem associated with the two horizon boundaries. We prescribe $z|_{H_1}=z_1$ and propagate the inverse lapse flow from $H_1$ into $\mathcal D$. Before reaching $H_{2}$, the monotonicity of $F(z)$ established in the main text gives
\begin{align}\label{mono}
    F(z_1)\leq F(z)\,,
\end{align}
along the leaf family generated from $H_1$.

For a general lapse distribution and boundary geometry, a flow starting from $H_1$ does not reach the other horizon section $H_{2}$. Instead, it may first reach only a portion of $H_2$ (see the red curve in Fig.~\ref{S3.4B} as a schematic diagram). Let $z_2$ denote the parameter value at which the evolving leaf first reaches $H_{2}$, and denote such ``first touched'' leaf by $\mathcal{S}_{z_2}^{(-)}$. Using the weighted-distance function $d_N(x,\mathcal{S}_{z_2})$ and the fact $N=0$ on $H_2$, one can find the $d_N(x,\mathcal{S}_{z_2})$ is constant for all $x\in H_2$. This means that the leaf $\mathcal{S}_{z_2}^{(-)}$ and $H_2$ have same weighted-distance, so that $H_{2}$ becomes a part of the level set $z=z_2$. Therefore, the level set $z=z_2$ consists of two parts: the leaf generated from $H_{1}$ is denoted by $\mathcal{S}_{z_2}^{(-)}$ and the horizon section $H_2$ is denoted by $\mathcal{S}_{z_2}^{(+)}$.

\begin{figure}
	\centering
	\includegraphics[width=5.6cm]{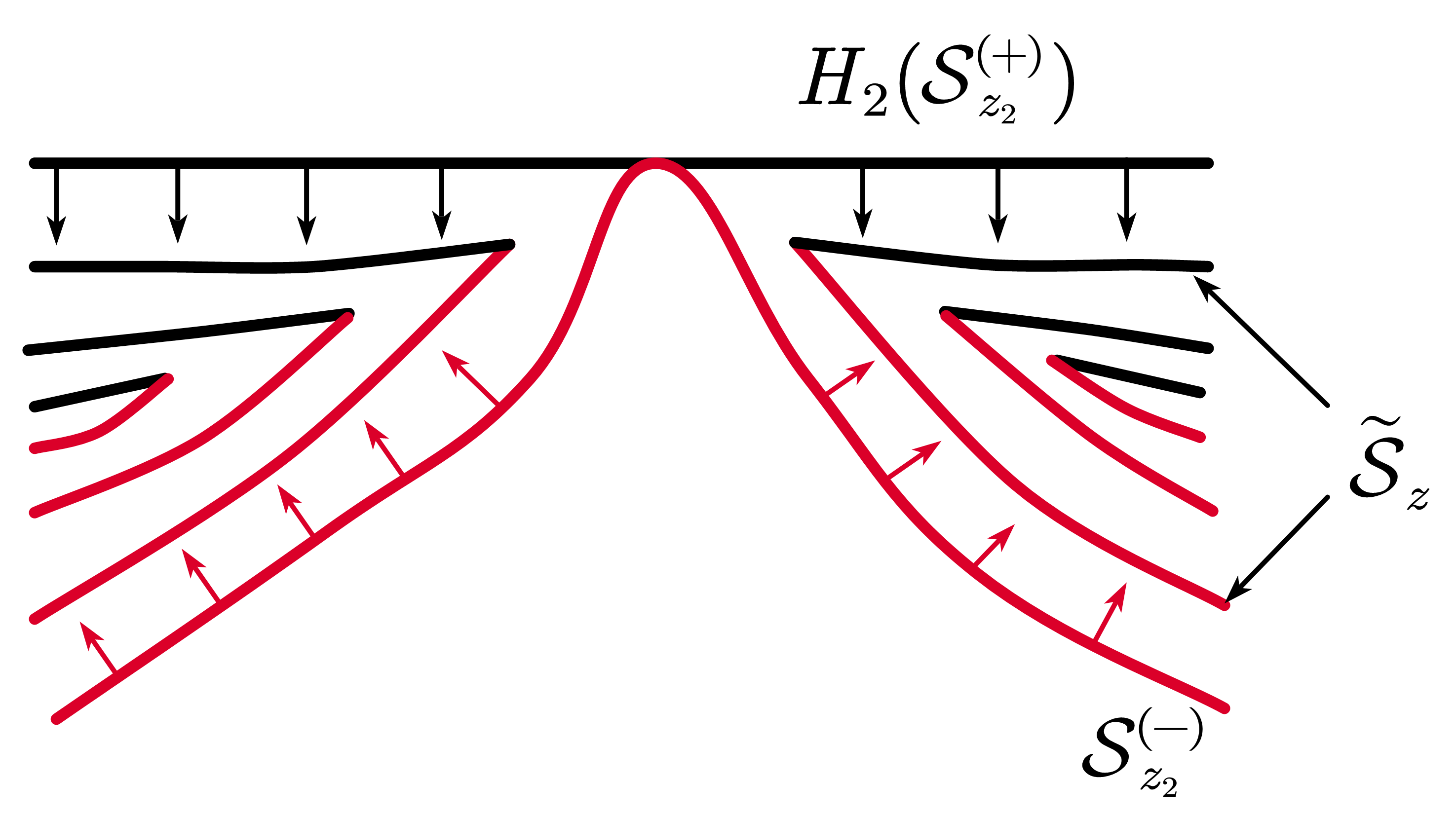}
    \captionsetup{
    justification=raggedright,
     singlelinecheck=false
    }
	\caption{Schematic illustration of the two-boundary overdetermination. At $z=z_2$, the weak level set consists of $\mathcal S_{z_2}^{(-)}$ generated from $H_1$ and $\mathcal S_{z_2}^{(+)}=H_2$. Under the viscosity continuation of the inverse lapse flow, the two branches of $\widetilde{\mathcal{S}}_{z}$ evolve toward each other and eventually merge at a ridge.}
	\label{S3.4B}
\end{figure}


Since the leaf $\mathcal{S}_{z_2}^{(-)}$ is generated from $H_{1}$ by the inverse lapse flow in the sense of viscosity solution, Eq.~\eqref{mono} gives 
\begin{equation}
    F|_{H_{1}}\le F|_{\mathcal{S}_{z_{2}}^{(-)}}\, .
\end{equation}
It remains to compare the leaf $\mathcal{S}_{z_2}^{(-)}$ with the other horizon section $H_2=\mathcal{S}_{z_2}^{(+)}$ since the $\mathcal{S}_{z_2}^{(+)}$ is just one smooth branch of the leaf at $z=z_2$. Let's consider the difference $F|_{\mathcal{S}_{z_{2}}^{(-)}}-F|_{\mathcal{S}_{z_{2}}^{(+)}}$. Denote the level set $z=z_2$ by $\widetilde S_{z_2}=\mathcal{S}_{z_2}^{(-)}\cup\mathcal{S}_{z_2}^{(+)}$. It needs to be noted that the sign of flux function $F$ depends on the orientation of the normal vector, and so there is a novel relationship 
\begin{equation}
    F|_{\widetilde {\mathcal S}_{z_2}}=F|_{\mathcal S_{z_2}^{(-)}} - F|_{\mathcal S_{z_2}^{(+)}}.
\end{equation}
To determine the sign of $F|_{\widetilde {\mathcal S}_{z_2}}$, we take $\widetilde S_{z_2}$ as the initial leaf and evolve it forward under the inverse lapse flow in the viscosity solution sense (see Fig.~\ref{S3.4B}). For $z \ge z_2$, the evolving weak leaves are denoted by
\begin{equation}
\widetilde{\mathcal S}_z = \mathcal S_z^{(-)} \cup \mathcal S_z^{(+)},
\end{equation}
where $\mathcal S_z^{(-)}$ generated from $\mathcal S_{z_2}^{(-)}$ and $\mathcal S_z^{(+)}$ generated from $\mathcal S_{z_2}^{(+)}$. 

As the leaf evolves, the two branches move toward each other forever until they eventually merge at a ridge where $z=z_r$ ($z_r$ could be finite or infinite). At this ridge, the two branches coincide geometrically but have opposite normals, so the flux on the ridge vanishes, i.e., $F|_{\widetilde{\mathcal S}_{z_{\mathrm r}}} = 0$. By the monotonicity of $F(z)$ established in the main text,
\begin{equation}
     F|_{\widetilde{\mathcal S}_{z_2}} \le F|_{\widetilde{\mathcal S}_{z_{\mathrm r}}} = 0.
\end{equation}
Hence, we obtain the difference
\begin{equation}
    F|_{\mathcal S_{z_2}^{(-)}} \le F|_{\mathcal S_{z_2}^{(+)}}.
\end{equation}
Combining this with the inequality Eq.~\eqref{mono} yields
\begin{equation}
     F|_{H_1} \le F|_{\mathcal S_{z_2}^{(-)}} \le F|_{\mathcal S_{z_2}^{(+)}} = F|_{H_2}.
\end{equation}
Therefore, the conclusion $F|_{H_1} \le F|_{H_2}$ remains valid even when the one-sided flow from $H_1$ fails to reach the entire second horizon $H_2$ simultaneously. The two-boundary overdetermination therefore does not obstruct the monotonicity argument.

\end{document}